\documentclass{llncs}
\usepackage{amsmath,amssymb,stmaryrd}
\usepackage[english]{babel}
\usepackage{hyperref}
\usepackage{float}
\usepackage[margin=1in]{geometry}
\def\Po#1#2{\Pos_{#1}{(#2)}}
\def\regw#1{\Reg{(#1)}}
\DeclareMathOperator{\Reg}{RegExp}
\DeclareMathOperator{\E}{E}

\DeclareMathOperator{\f}{F}
\DeclareMathOperator{\G}{G}
\DeclareMathOperator{\First}{First}
\DeclareMathOperator{\Last}{Last}
\DeclareMathOperator{\Follow}{Follow}
\DeclareMathOperator{\rooot}{root}
\DeclareMathOperator{\Pos}{Pos}
\DeclareMathOperator{\Leave}{Leaves}
\def\b#1{\overline{#1}}

\usepackage{tikz}
  \usetikzlibrary{automata}
  \usetikzlibrary{shadows}
  \usetikzlibrary{arrows}
  \usetikzlibrary{shapes}
  \usetikzlibrary{decorations.pathmorphing}
  \usetikzlibrary{fit}
  \usetikzlibrary{arrows}

\tikzstyle{every picture}=[>=stealth',shorten >=1pt,node distance=1.44cm,bend angle=45,initial text=,every state/.style={inner sep=0.75mm, minimum size=1mm},font=\scriptsize]
\begin{document}

\title{$K$-Position, Follow, Equation and $K$-C-Continuation Tree Automata Constructions} 
 
  \author{
    Ludovic Mignot, Nadia Ouali Sebti and Djelloul Ziadi\thanks{D.~Ziadi was supported by the MESRS - Algeria under Project 8/U03/7015.}
  } 
  
  \institute{
    LITIS, Universit\'e de Rouen, 76801 Saint-\'Etienne du Rouvray Cedex, France\\
     \email{\{ludovic.mignot,nadia.ouali-sebti,djelloul.ziadi\}@univ-rouen.fr}
  }
  
  \maketitle  
  
  \begin{abstract} 
   There exist several methods of computing an automaton recognizing the language denoted by a given regular expression: 
  In the case of words, 
  the \emph{position automaton} $\mathcal{P}$ due to Glushkov, the \emph{c-continuation automaton} $\mathcal{C}$ due to Champarnaud and Ziadi, the \emph{follow automaton} $\mathcal{F}$ due to 
  Ilie
  and Yu and the equation automaton $\mathcal{E}$ due to Antimirov.  
  It has been shown that $\mathcal{P}$ and $\mathcal{C}$ are isomorphic and that $\mathcal{E}$ (resp. $\mathcal{F}$) is a quotient of $\mathcal{C}$ (resp. of $\mathcal{P}$).
  
  In this paper, we define from a given regular tree expression the $k$-position tree automaton $\mathcal{P}$ and the follow tree automaton $\mathcal{F}$. 
  Using the definition of the equation tree automaton $\mathcal{E}$ of Kuske and Meinecke and our previously defined $k$-C-continuation tree automaton $\mathcal{C}$, we show that the previous morphic relations are still valid 
  on tree expressions.  
  \end{abstract}
  
\section{Introduction}
  Regular expressions are used in numerous domains of applications in computer science. 
  They are an easy and compact way to represent potentially infinite regular languages, that are well-studied objects leading to efficient decision problems.
  Among them, the membership test, that is to determine whether or not a given word belongs to a language.
  Given a regular expression $E$ with $n$ symbols and a word $w$, to determine whether $w$ is in the language denoted by $E$ can be polynomially performed (with respect to $n$) \emph{via} the computation of a finite state machine, called an automaton, that can be seen as a symbol-labelled graph with initial and final states.
  There exist several methods to compute such an automaton.
  
  The first approach is to determine particular properties over the syntactic structure of the regular expression $E$. 
  Glushkov~\cite{glushkov} proposed the computation of four position functions $\mathrm{Null}$, $\mathrm{First}$, $\mathrm{Last}$, and $\mathrm{Follow}$, which once computed, lead to the computation of a $(n+1)$-
  state
   automaton. 
  Ilie
   and Yu showed in~\cite{yu} how to reduce it by merging similar states.  
  Another method is to compute the transition function of the automaton as follows: 
  associating a regular expression with a state $s$,
  any path labelled by a word $w$ brings the automaton from 
  the
  state $s$ into a finite set of states $S'=\{s'_1,\ldots,s'_k\}$ such that these states denote the quotient $w^{-1}(L(s))$ of the language $L(s)$ by $w$, that contains the word $w'$ such that $ww'$ belongs to $L(s)$.
  Basically, it is a computation that tries to determine what words $w'$ can be accepted after reading a prefix $w$. 
  The first author that introduced such a process is Brzozowski~\cite{brzozowski}. 
  He showed how to compute a regular expression denoting $w^{-1}(L(E))$ from the expression $E$: this expression, denoted by $d_{w}(E)$, is called the \emph{derivative} of $E$ with respect to $w$.
  Furthermore, the set of dissimilar derivatives, combined with reduction according to associativity, commutativity and idempotence of the sum, is finite and can lead to the computation of a deterministic finite automaton.
  Antimirov \cite{antimirov} extended this method to the computation of partial derivatives, that are no longer expressions but sets of expressions. These so-called derived terms produce the \emph{equation automaton}.
  Finally, by deriving expressions after having them indexed, Champarnaud and Ziadi~\cite{ZPC2} computed the \emph{c-continuation automaton}.
  
  The different morphic links between these four automata have been studied too: 
  Ilie
   and Yu showed that the follow automaton is a quotient of the position automaton; Champarnaud and Ziadi proved that the position automaton and the c-continuation 
  automaton
  are isomorphic and that the equation automaton is a quotient of the position 
  automaton.   
  Finally, using a join of the two previously defined 
  quotients,
   Garcia \emph{et al.} presented in~\cite{garcia} an automaton that is smaller than both 
   the follow and the equation automata.
  
  In this paper, we extend the study of these morphic links to different computations of tree automata. We define two new tree automata constructions, \emph{the $k$-position automaton} and 
  \emph{the follow automaton},
   and we study their morphic links with two other already known automata constructions, the \emph{equation automaton} of Kuske and Meinecke~\cite{automate2} and our $k$\emph{-C-continuation automaton}~\cite{cie}.   
  Notice that a position automaton and a reduced 
  automaton
   have already been defined in~\cite{Ouali}. However, they are not isomorphic with the automata we define in this paper.
  This study is motivated by the development of a library of functions for handling rational kernels \cite{mohri1} in the case of trees. 
  The first problem consists in converting a regular tree expression into a tree transducer.  
  Section~\ref{sec prelim} recalls basic definitions and properties of regular tree languages and regular tree expressions. 
  In Section~\ref{sec tree automata}, we define two new automata computations, \emph{the $k$-position automaton} and 
  \emph{the follow automaton}
   and recall the definition of the \emph{equation automaton} and of the $k$\emph{-C-continuation automaton}; we also present the morphic links between these four methods in this section.
  Section~\ref{sec compar} is devoted to the comparison of the follow automaton and of the equation 
  automaton;
  it is proved that there are no morphic link between them.
  Moreover, we extend the computation of the Garcia \emph{et al.} equivalence leading to a smaller automaton in this section.

\section{Preliminaries}\label{sec prelim}
    Let $(\Sigma,\mathrm{ar})$ be  \emph{a ranked alphabet}, where $\Sigma$ is a finite set and $\mathrm{ar}$ represents the  \emph{rank} of $\Sigma$ which is a mapping from $\Sigma$ into $\mathbb{N}$. The set of symbols of rank $n$ is denoted by $\Sigma_{n}$. The elements of rank $0$ are called  \emph{constants}. A \emph{tree} $t$ over   $\Sigma$ is inductively defined as follows: $t=a,~ t=f(t_1,\dots,t_k)$ where $a$ is any symbol in  $\Sigma_0$, $k$ is any integer satisfying $k\geq 1$, $f$ is any symbol in $\Sigma_k$ and $t_1,\dots,t_k$ are any $k$ trees over $\Sigma$. We denote by $T_{\Sigma}$ the set of trees over $\Sigma$.  \emph{A tree language} is a subset of $T_{\Sigma}$. Let ${\Sigma}_{\geq 1}=\Sigma\backslash \Sigma_0$ denote the set of  \emph{non-constant symbols} of the ranked alphabet $\Sigma$. \emph{A Finite Tree Automaton} (FTA)~\cite{automate1,automate2} ${\cal A}$ is a tuple $( Q, \Sigma, Q_{T},\Delta )$ where $Q$ is a finite set of states, $Q_T \subset Q$ is the set of \emph{final states} 
and  $\Delta\subset\bigcup_{n\geq 0}(Q \times \Sigma_{n}\times Q^n)$ is the set of  \emph{transition rules}. This set is equivalent to the function $\Delta$ 
from $Q^n \times \Sigma_n$ to $2^Q$
defined  by $(q,f,q_1,\dots,q_n)\in \Delta\Leftrightarrow q\in \Delta(q_1,\dots,q_n,f)$. The domain of this function can be extended to  
$(2^Q)^n \times \Sigma_n$
as follows: $\Delta(Q_1,\dots,Q_n,f)=\bigcup_{(q_1,\dots,q_n)\in Q_1\times\dots\times Q_n} \Delta(q_1,\dots,q_n,f)$.  Finally, we denote by $\Delta^*$ the function from  $T_{\Sigma}\rightarrow 2^Q$  defined for any tree in $T_{\Sigma}$ as follows: 
  $\Delta^*(t)= \Delta(a)$ if  $t=a$ with $a\in \Sigma_0$, 
  $\Delta^*(t)=\Delta(\Delta^*(t_1),\dots,\Delta^*(t_n),f)$
   if $t=f(t_1,\dots,t_n)$ with $f\in {\Sigma}_n$ and $t_1,\ldots,t_n\in T_{\Sigma}$. 
  A tree is \emph{accepted} by ${\cal A}$ if and only if $\Delta^*(t)\cap Q_T\neq \emptyset$.   
  The language ${\cal L(A)}$ \emph{recognized} by $A$ is the set of trees accepted by ${\cal A}$  \emph{i.e.} ${\cal L(A)}=\{t\in T_{\Sigma}\mid \Delta^*(t)\cap Q_T\neq \emptyset\}$. 
 Let $\sim$ be an equivalence relation over $Q$. We denote by $[q]$ the equivalence class of any state $q$ in $Q$. The \emph{quotient of} $A$ w.r.t. $\sim$ is the tree automaton $A_{/\sim}=( Q_{/\sim}, \Sigma, {Q_{T}}_{/\sim},\Delta_{/\sim} )$ where: $Q_{/\sim}=\{[q]\mid q\in Q\}$, ${Q_{T}}_{/\sim}=\{[q]\mid [q]\cap Q_T \neq \emptyset\}$, $\Delta_{/\sim}=\{([q],f,[q_1],\ldots,[q_n]) \mid ([q]\times \{f\}\times [q_1]\times \ldots\times [q_n])\cap \Delta\neq \emptyset\}$.
Notice that a transition $([q],f,[q_1],\ldots,[q_n])$ in $\Delta_{/\sim}$ does not imply a transition $(q,f,q_1,\ldots,q_n)$ in $\Delta$. Moreover, the relation $\sim$ is not necessarily a congruence w.r.t. the transition function: in this paper, we will deal with specific equivalence relations (similarity relations) that turn to be congruences. This particular considerations will be clarified in Subsection~\ref{subsec folltreeaut}.

  For any integer $n\geq 0$, for any $n$ languages $L_1, \dots, L_n\subset T_{\Sigma}$, and for any symbol  $f\in \Sigma_n$, $f(L_1, \dots, L_n)$ is the tree language $\lbrace f(t_1, \dots, t_n)\mid t_i\in L_i\rbrace$. The \emph{tree substitution} of a constant $c$ in $\Sigma$ by a language $L\subset T_{\Sigma}$ in a tree $t\in T_{\Sigma}$, denoted by $t\lbrace c \leftarrow L\rbrace$, is the language inductively 
  defined by:
   $L$ if $t=c$; $\lbrace d\rbrace$ if $t=d$ where $d\in \Sigma_0\setminus\{c\}$; $f(t_1\lbrace c \leftarrow L\rbrace, \dots, t_n\lbrace c \leftarrow L\rbrace)$ if $t=f(t_1, \dots, t_n)$ with $f\in\Sigma_n$ and $t_1, \dots, t_n$ any $n$ trees over $\Sigma$.   
    Let $c$ be a symbol in $\Sigma_0$. The $c$-\emph{product} $L_1\cdot_{c} L_2$ of two languages $L_1, L_2\subset T_{\Sigma}$ is  defined by $L_1\cdot_{c} L_2=\bigcup_{t\in L_1}\lbrace t\lbrace c \leftarrow L_2\rbrace \rbrace$. The \emph{iterated $c$-product} is  inductively  defined for $L\subset T_{\Sigma}$ by:  $L^{0_c}=\lbrace c \rbrace$ and $L^{{(n+1)}_c}=L^{n_c}\cup L\cdot_{c} L^{n_c}$. The $c$-\emph{closure} of $L$ is defined by  $L^{*_c}=\bigcup_{n\geq 0} L^{n_c}$.
    
A \emph{regular expression} over a ranked alphabet $\Sigma$ is inductively defined by  $\E=0$, $\E\in \Sigma_0$, $\E=f(\E_1, \cdots, \E_n)$, $\E=(\E_1+\E_2)$, $\E=(\E_1\cdot_c \E_2)$, $\E=({\E_1}^{*_c})$, where $c\in\Sigma_0$, $n\in\mathbb{N}$, $f\in\Sigma_n$ and $\E_1,\E_2 ,\dots, \E_n$ are any $n$ regular expressions over $\Sigma$. Parenthesis can be omitted when there is no ambiguity. We write $\E_1=\E_2$ if $\E_1$ and $\E_2$ graphically coincide. We denote by $\regw{\Sigma}$ the set of all regular expressions over $\Sigma$. Every regular expression $\E$ can be seen as a tree over the ranked alphabet 
$\Sigma\cup \{+,\cdot_c, *_c \mid c\in \Sigma_0\}$
 where $+$ and $\cdot_c$ can be seen 
 as symbols
 of rank $2$ and $*_c$ has rank $1$. This tree is the syntax-tree $T_{\E}$ of $\E$. The \emph{alphabetical width} $||\E||$ of $\E$ is the number of occurrences of symbols of $\Sigma$ in $\E$. \emph{The size} $|\E|$ of $\E$ is the size of its syntax tree $T_{\E}$. The \emph{language} $\llbracket \E\rrbracket$  \emph{
denoted by} $\E$ is 
inductively defined by
 $\llbracket 0\rrbracket=\emptyset$, $\llbracket c\rrbracket=\lbrace c\rbrace$, $\llbracket f(\E_1,\E_2 , \cdots, \E_n)\rrbracket= f(\llbracket \E_1 \rrbracket, \dots, \llbracket \E_n \rrbracket)$, $\llbracket \E_1+ \E_2\rrbracket=\llbracket \E_1\rrbracket\cup\llbracket \E_2 \rrbracket$, $\llbracket \E_1\cdot_{c} \E_2\rrbracket=\llbracket \E_1\rrbracket \cdot_{c}\llbracket \E_2 \rrbracket$, $\llbracket {\E_1}^{*_c}\rrbracket=\llbracket \E_1\rrbracket^{*_c}$ where  $n\in\mathbb{N}$, $\E_1,\E_2,\dots, \E_n$ are any $n$ regular expressions, $f\in\Sigma_n$  and $c\in \Sigma_0$. It is well known that a tree language  is accepted by some tree automaton if and only if it can be denoted by a regular expression \cite{automate1,automate2}.
A regular expression $\E$ defined over $\Sigma$ is  \emph{linear} 
if
every symbol of 
rank greater than $1$
 appears at most once in $\E$. Note that any constant symbol may occur more than once. Let $\E$ be a regular expression over $\Sigma$. The  \emph{linearized regular expression} $\b\E$ in $\E$ of a regular expression $\E$ is obtained from $\E$ by marking differently all symbols of a rank  greater than or equal to $1$ (symbols of $\Sigma_{\geq 1}$). 
The marked symbols form together with the constants in $\Sigma_0$ a ranked alphabet $\mathrm{Pos}_E(E)$ the symbols of which we call \emph{positions}.
The mapping $h$ is defined from $\Po{\E}{\E}$ to $\Sigma$ with $h(\Po{\E}{\E}_m)\subset \Sigma_m$ for every  $m\in \mathbb{N}$. It associates with a marked symbol $f_j\in  \Po{\E}{\E}_{\geq 1}$ the symbol $f\in \Sigma_{\geq 1}$ and for a symbol $c\in \Sigma_0$ the symbol $h(c)=c$.
We can extend the mapping $h$ naturally to  $\regw{\Po{\E}{\E}}\rightarrow\regw{\Sigma}$ by $h(a)=a$, $h(\E_1+\E_2)=h(\E_1)+h(\E_2)$, $h(\E_1\cdot_c\E_2)=h(\E_1)\cdot_c h(\E_2)$, $h(\E_1^{*_c})=h(\E_1)^{*_c}$, $h(f_j(\E_1,\dots,\E_n))=f(h(\E_1),\dots,h(\E_n))$, with $n\in\mathbb{N}$, $a\in \Sigma_0$, $f\in \Sigma_n$, $f_j\in \Po{\E}{\E}_n$ such that $h(f_j)=f$ and $\E_1,\dots,\E_n$ any regular expressions over $\Po{\E}{\E}$. 


\section{Tree Automata from Regular Expressions}\label{sec tree automata}

  In this section, we show how to compute from a regular expression $\E$ four tree automata accepting $\llbracket \E\rrbracket$: we introduce two new constructions, the $K$-position automaton and the follow automaton of $\E$, and then we recall two already-known constructions, the equation automaton~\cite{automate2} and the C-continuation automaton~\cite{cie}.
  
  Regular languages defined over ranked alphabet $\Sigma$ are exactly the languages denoted by a regular expression on $\Sigma$. 
  There may exist many distinct regular expressions which denote the same regular language. 
  Two regular expressions are said to be \emph{equivalent} if they denote the same language.
 In what follows we  only consider expressions without $0$ or reduced to $0$.
 
  In the following of this section, $\E$ is a regular expression over a ranked alphabet $\Sigma$. 
  The set of symbols in $\Sigma$ that appear in an expression $\f$ is denoted by $\Sigma_F$.

  \subsection{The $K$-Position Tree Automaton}

  In this section, we show how to compute the $K$-position tree automaton of a regular expression $E$, recognizing $\llbracket E\rrbracket$. This is an extension of the well-known position automaton~\cite{glushkov} for word regular expressions
  where the $K$ represents the fact that any $k$-ary symbol is no longer a state of the automaton, but is exploded into $k$ states.
  The same method was presented independently by McNaughton and Yamada~\cite{mcnaughton60}.  
   Its computation is based on the computations of particular \emph{position functions}, defined in the following.
  
  In what follows, for any two trees $s$ and $t$, we denote by $s\preccurlyeq t$ the relation "$s$ is a subtree of $t$". 
  Let
  $t=f(t_1,\dots,t_n)$ be a tree. 
We denote by $\rooot(t)$ the root of $t$, by $k\mbox{-}\mathrm{child(t)}$ the $k^{th}$ child of $f$ in $t$, that is the root of $t_k$ if it exists, and by $\Leave(t)$ the set of the leaves of $t$, \emph{i.e.} $\{s\in \Sigma_0\mid s\preccurlyeq t\}$. 
  
  Let $\E$ be linear, $1\leq k\leq m$ be two integers and $f$ be a symbol in $\Sigma_m$.
  %
  The set $\First(\E)$ is the subset of $\Sigma$ defined by $\{\rooot(t)\in \Sigma \mid t\in \llbracket\E \rrbracket\}$; The set $\Follow(\E,f,k)$ is the subset of $\Sigma$ defined by $\{g\in \Sigma \mid \exists t\in \llbracket \E \rrbracket, \exists s\preccurlyeq t, \mathrm{root}(s)=f, k\mbox{-}\mathrm{child(s)}=g\}$; The set $\Last(\E)$ is the subset of $\Sigma_0$ defined by $\Last(\E)=\displaystyle\bigcup_{t\in\llbracket \E\rrbracket}\Leave(t)$.
  
 \begin{example}\label{Pos Automat}
 Let $\Sigma=\Sigma_0\cup\Sigma_1\cup\Sigma_2$ be defined by $\Sigma_0=\{a,b,c\}$, $\Sigma_1=\{f,h\}$ and $\Sigma_2=\{g\}$.
    Let us consider the regular expression $\E$ and its linearized form defined by:
    
    \centerline{$\E=(f(a)^{*_a}\cdot_a b+ h(b))^{*_b}+g(c,a)^{*_c}\cdot_c (f(a)^{*_a}\cdot_a b+ h(b))^{*_b}$,}
    
    \centerline{$\b\E=(f_1(a)^{*_a}\cdot_a b+ h_2(b))^{*_b}+g_3(c,a)^{*_c}\cdot_c (f_4(a)^{*_a}\cdot_a b+ h_5(b))^{*_b}$.}

The language denoted by $\b\E$ is 

\centerline{
  \begin{tabular}{l@{\ }l}
    $\llbracket \b\E\rrbracket=$ & $\{b, f_1(b),f_1(f_1(b)),f_1(h_2(b)),h_2(b),h_2(f_1(b)),h_2(h_2(b)), \ldots,$\\
    & $g_3(b,a),
g_3(g_3(b,a),a),
 g_3(f_4(b),a),g_3(h_5(b),a),$\\
 & $f_4(f_4(b)),f_4(h_5(b),h_5(f_4(b)),h_5(h_5(b)),\ldots\}$\\
   \end{tabular}
}

Consequently, $\First(\b\E)=\{b,f_1,h_2,g_3,f_4,h_5\}$ and $\Follow(\b\E,f_1,1)=\{b,f_1,h_2\}$, $\Follow(\b\E,h_2,1)=\{b,f_1,h_2\}$, $\Follow(\b\E,g_3,1)=\{b,g_3,f_4,h_5\}$, $\Follow(\b\E,g_3,2)=\{a\}$, $\Follow(\b\E,f_4,1)=\{b,f_4,h_5\}$,\\ $\Follow(\b\E,h_5,1)=\{b,f_4,h_5\}$. 
\end{example}  

   Let us first show that the position functions $\mathrm{First}$ and $\mathrm{Follow}$ are    
   inductively computable.

  \begin{lemma}\label{firstComput}
    Let $\E$ be linear.
    The set $\First(\E)$ can be computed 
    as follows:
    
    \centerline{$\First(0)=\emptyset$, $\First(a)=\lbrace a\rbrace$, $\First(f(\E_1, \cdots,\E_m))=\lbrace f\rbrace$,}
    
    \centerline{$\First(\E_1+\E_2)=\First(\E_1)\cup\First(\E_2)$, $\First({\E_1}^{*_c})=\First(\E_1)\cup\lbrace c\rbrace$,}
    
    \centerline{
      $\First(\E_1\cdot_c \E_2)=
        \left\{
          \begin{array}{l@{\ }l}
            (\First(\E_1)\setminus\{c\}) \cup \First(\E_2) & \text{ if } c\in\llbracket \E_1\rrbracket,\\
            \First(\E_1) & \text{ otherwise.}\\ 
          \end{array}
        \right.
      $
    }
  \end{lemma}  
  \begin{proof}
    Let us show by induction over $E$ that any symbol $f$ in $\Sigma$ belongs to $\mathrm{First}(E)$ if and only if there exists a tree $t$ in $\llbracket E\rrbracket$ the root of which is $f$ (\emph{i.e.} $t=f(t_1,\ldots,t_n)$).
    \begin{enumerate}
      \item If $E=0$, $\llbracket E\rrbracket=\emptyset=\First(E)$.
      \item Let us suppose that $E=a$ with $a\in\Sigma_0$. Hence $\llbracket E\rrbracket=\{a\}=\First(E)$.
      \item Suppose that $E=f(E_1,\ldots,E_n)$. 
        Hence $\llbracket E\rrbracket=f(\llbracket E_1\rrbracket,\ldots,\llbracket E_n\rrbracket)$. 
        As a direct consequence, any tree $t$ in $\llbracket E\rrbracket$ admits $f$ as root. 
        Consequently, $\{f\}=\mathrm{First}(E)$ and  $\exists f(t_1,\ldots,t_n) \in \llbracket E\rrbracket$.
      \item Suppose that $E=E_1+E_2$.
        Then:
        
        \centerline{
          \begin{tabular}{l@{\ }l}
            & $f\in \mathrm{First}(E_1)\cup \mathrm{First}(E_2)$\\
            $\Leftrightarrow$ & $f\in \mathrm{First}(E_1)\vee f\in \mathrm{First}(E_2)$\\
            $\Leftrightarrow$ & $\exists f(t_1,\ldots,t_n) \in \llbracket E_1\rrbracket \vee \exists f(t_1,\ldots,t_n) \in \llbracket E_2\rrbracket$\ \ \ \ \textbf{(Induction Hypothesis)}\\
            $\Leftrightarrow$ & $\exists f(t_1,\ldots,t_n) \in \llbracket E_1\rrbracket \cup \llbracket E_2\rrbracket$\\
            $\Leftrightarrow$ & $\exists f(t_1,\ldots,t_n) \in \llbracket E\rrbracket$\\
          \end{tabular}
        }
    
      \item Let us consider that $E=E_1\cdot_c E_2$, with $c\in\Sigma_0$.
        \begin{enumerate}
          \item\label{item concat first comput} If $c\notin \llbracket E_1\rrbracket$, then $\mathrm{First}(E)=\mathrm{First}(E_1)$. 
            
            \noindent Let $t=f(t_1,\ldots,t_n)$ be a tree in $\llbracket E_1\rrbracket\setminus\{c\}$. 
            Since $\llbracket E_2\rrbracket\neq\emptyset$, $t_{c\leftarrow \llbracket E_2\rrbracket}\neq\emptyset$. 
            Furthermore, any tree $t'$ in $t_{c\leftarrow \llbracket E_2\rrbracket}\neq\emptyset$ admits the same root as $t$ by construction.
            Hence $\exists t=f(t_1,\ldots,t_n) \in \llbracket E_1\rrbracket$ $\Rightarrow$ $\exists t'=f(t'_1,\ldots,t'_n) \in \llbracket E_1\cdot_c E_2\rrbracket$
            
            \noindent Reciprocally, let $t'=f(t'_1,\ldots,t'_n)$ be a tree in $\llbracket E_1\cdot_c E_2\rrbracket$. 
            By definition, there exists a tree $t$ in $\llbracket E_1\rrbracket$ such that $t'\in t_{c\leftarrow \llbracket E_2\rrbracket}$.
            Trivially, since $t\neq c$, $t$ and $t'$ admit the same root.
            Hence $\exists t=f(t_1,\ldots,t_n) \in \llbracket E_1\rrbracket$ $\Leftarrow$ $\exists t'=f(t'_1,\ldots,t'_n) \in \llbracket E_1\cdot_c E_2\rrbracket$.
            
            \noindent Consequently, the following condition \textbf{(C1)} holds:
            
            \centerline{
              $\exists t=f(t_1,\ldots,t_n) \in \llbracket E_1\rrbracket\setminus\{c\}$ $\Leftrightarrow$ $\exists t'=f(t'_1,\ldots,t'_n) \in \llbracket E_1\rrbracket\setminus\{c\}\cdot_c \llbracket E_2\rrbracket$.
            }
            
            \noindent Finally:
            
            \centerline{
              \begin{tabular}{l@{\ }l}
                $f\in \mathrm{First}(E)$ & $\Leftrightarrow$ $f\in \mathrm{First}(E_1)$\\
                & $\Leftrightarrow$  $\exists t=f(t_1,\ldots,t_n) \in \llbracket E_1\rrbracket$\ \ \ \ \textbf{(Induction hypothesis)}\\
                & $\Leftrightarrow$ $\exists t'=f(t'_1,\ldots,t'_n) \in \llbracket E_1\rrbracket\setminus\{c\}\cdot_c \llbracket E_2\rrbracket$\ \ \ \ \textbf{(C1)}\\
                & $\Leftrightarrow$ $\exists t'=f(t'_1,\ldots,t'_n) \in \llbracket E\rrbracket$\\
              \end{tabular}
            }
            
          \item Consider that $c\in \llbracket E_1\rrbracket$. 
            Hence $\mathrm{First}(E)=\mathrm{First}(E_1)\setminus\{c\} \cup \mathrm{First}(E_2)$. 
            Since $E\neq 0$ $\Rightarrow$ $E_2\neq 0$ $\Rightarrow$ $\llbracket E_2\rrbracket\neq\emptyset$, let $u_2$ be a tree in $\llbracket E_2\rrbracket$. 
            Similarly, since $E\neq 0$ $\Rightarrow$ $E_1\neq 0$ $\Rightarrow$ $\llbracket E_1\rrbracket\neq\emptyset$, let $u_1$ be a tree in $\llbracket E_1\rrbracket$. 
            
            \noindent According to condition \textbf{(C1)} in case~\ref{item concat first comput},
            
            \centerline{
              $\exists t=f(t_1,\ldots,t_n) \in \llbracket E_1\rrbracket\setminus\{c\}$ $\Leftrightarrow$ $\exists t'=f(t'_1,\ldots,t'_n) \in \llbracket E_1\rrbracket\setminus\{c\}\cdot_c \llbracket E_2\rrbracket$.
            }
            
            \noindent Furthermore, by definition of $\cdot_c$, the following condition \textbf{(C2)} holds:
            
            \centerline{
              $\exists f(t_1,\ldots,t_n) \in \llbracket E_2\rrbracket$ $\Leftrightarrow$ $\exists f(t_1,\ldots,t_n) \in \{c\} \cdot_c \llbracket E_2\rrbracket$
            }
            
            \noindent Consequently:
            
            \centerline{
              \begin{tabular}{l@{\ }l}
                & $f\in \mathrm{First}(E)$\\
                $\Leftrightarrow$  & $f\in \mathrm{First}(E_1)\setminus\{c\} \cup \mathrm{First}(E_2)$\\
                $\Leftrightarrow$  & $\ \ \exists f(t_1,\ldots,t_n) \in \llbracket E_1\rrbracket\setminus\{c\}$\\
                & $\vee \exists f(t_1,\ldots,t_n) \in \llbracket E_2\rrbracket$\ \ \ \ \textbf{(Induction hypothesis)}\\
                $\Leftrightarrow$  & $\ \ \exists f(t_1,\ldots,t_n) \in \llbracket E_1\rrbracket\setminus\{c\}\cdot_c \llbracket E_2\rrbracket$\\
                & $ \vee \exists f(t_1,\ldots,t_n) \in \{c\}\cdot_c\llbracket E_2\rrbracket$\ \ \ \ \textbf{(C1) and (C2)}\\
                $\Leftrightarrow$  & $\exists f(t_1,\ldots,t_n) \in \llbracket E_1\rrbracket\cdot_c \llbracket E_2\rrbracket$\\
                $\Leftrightarrow$  & $\exists f(t_1,\ldots,t_n) \in \llbracket E\rrbracket$\\
              \end{tabular}
            }
        \end{enumerate}
        
      \item Let us suppose that $E=E_1^{*_c}$. 
        Then $\mathrm{First}(E)=\mathrm{First}(E_1)\cup\{c\}$.
        Let $f$ be a symbol in $\Sigma$. If $f=c$, then $c\in \mathrm{First}(E)$ and $c\in \llbracket E\rrbracket$.
        Otherwise, since $E\neq 0$ $\Rightarrow$ $E_1\neq 0$ $\Rightarrow$ $\llbracket E_1\rrbracket\neq\emptyset$, there exists some tree in $\llbracket E_1\rrbracket$.
        According to condition \textbf{(C1)} in case~\ref{item concat first comput},
        
        \centerline{
          $\exists t=f(t_1,\ldots,t_n) \in \llbracket E_1\rrbracket\setminus\{c\}$ $\Leftrightarrow$ $\exists t'=f(t'_1,\ldots,t'_n) \in \llbracket E_1\rrbracket\setminus\{c\}\cdot_c \llbracket E_1\rrbracket$.
        }
        
        \noindent By successive application of \textbf{(C1)}, it can be shown that the following condition \textbf{(C3)} holds: for any integer $n\geq 1$, 
        
        \centerline{
          $\exists t=f(t_1,\ldots,t_n) \in \llbracket E_1\rrbracket\setminus\{c\}$ $\Leftrightarrow$ $\exists t'=f(t'_1,\ldots,t'_n) \in \llbracket E_1\rrbracket^{n_c}$.
        }
        
        \noindent Consequently:
        
        \centerline{
          \begin{tabular}{l@{\ }l}
            $f\in \mathrm{First}(E)$ & $\Leftrightarrow$ $f\in \mathrm{First}(E_1)\setminus \{c\}$\\
            & $\Leftrightarrow$  $\exists f(t_1,\ldots,t_n) \in \llbracket E_1\rrbracket\setminus\{c\}$ \hfill \textbf{(Induction hypothesis)}\\
            & $\Leftrightarrow$  $\exists n\geq 1, \exists t'=f(t'_1,\ldots,t'_n) \in \llbracket E_1\rrbracket^{n_c}$ \\
            & $\Leftrightarrow$  $\exists f(t_1,\ldots,t_n)\in \llbracket E_1^{*_c}\rrbracket$\\
          \end{tabular}
        }
    \end{enumerate}
  \end{proof}

 \begin{lemma}\label{Computfollow} 
    Let $\E$ be linear, $1\leq k\leq m$ be two integers and $f$ be a symbol in $\Sigma_m$.  
    The set of symbols $\Follow(\E,f,k)$ can be computed inductively as follows:
    
    \centerline{$\Follow(0,f,k)=\Follow(a,f,k)=\emptyset$,}
    
    \centerline{
      $\Follow(g(\E_1,\ldots,\E_n),f,k)=
        \left\{
          \begin{array}{l@{\ }l}
            \First(\E_k) & \text{ if } f=g,\\
            \Follow(\E_l,f,k) & \text{ if } \exists l\mid f\in \Sigma_{E_l},\\
            \emptyset & \text{ otherwise }. 
          \end{array}
        \right.$
    }
    
    \centerline{
      $\Follow(\E_1+\E_2,f,k)=
        \left\{
          \begin{array}{l@{\ }l}
            \Follow(\E_1,f,k) & \text{ if } f\in \Sigma_{\E_1},\\
            \Follow(\E_2,f,k) & \text{ if } f\in \Sigma_{\E_2},\\ 
             \emptyset & \text{ otherwise }. 
          \end{array}
        \right.$
    }
    
    \centerline{
      $\Follow(\E_1 \cdot_c \E_2,f,k)=
        \left\{
          \begin{array}{l@{\ }l}
            (\Follow(\E_1,f,k)\setminus\{c\}) \cup \First(\E_2) & \text{ if } c\in\Follow(\E_1,f,k),\\
            \Follow(\E_1,f,k) & \text{ if } f\in \Sigma_{\E_1}\\
            &\ \  \wedge c\notin\Follow(\E_1,f,k),\\
            \Follow(\E_2,f,k) & \text{ if } f\in \Sigma_{\E_2}\\
            & \ \  \wedge c\in\Last(\E_1),\\
            \emptyset & \text{ otherwise,}
          \end{array}
        \right.$
    }
    
    \centerline{
      $\Follow(\E_1^{*_c},f,k)=
        \left\{
          \begin{array}{l@{\ }l}
            \Follow(\E_1,f,k) \cup \First(\E_1) & \text{ if } c\in\Follow(\E_1,f,k),\\
            \Follow(\E_1,f,k) & \text{ otherwise,}\\ 
          \end{array}
        \right.$
    }
  \end{lemma} 
  \begin{proof}
    Let $F$ be any linear regular expression, $h$ be any symbol in $\Sigma$, and $j$ be any integer. 
    We denote by $\phi^{h,j}_F$ the set $\{g\in \Sigma \mid \exists t\in \llbracket F \rrbracket, \exists s\preccurlyeq t, \mathrm{root}(s)=h\wedge j\mbox{-}\mathrm{child(s)}=g\}$.
    Let us show by induction over the structure of $E$ that $\mathrm{Follow}(E,f,k)=\phi^{f,k}_E$. 
    Let $h$ be a symbol in $\Sigma$. 
    \begin{enumerate}
      \item If $E=0$ or if $E=a$, then $\mathrm{Follow}(E,f,k)=\emptyset=\phi^{f,k}_E$.
      \item Suppose that $E=g(E_1,\ldots,E_n)$. 
        Since $E\neq 0$, then for any $1\leq j\leq n$, $E_j\neq 0$ and then there exists at least one tree $u_j$ in $\llbracket E_j \rrbracket$. 
        Three cases have to be considered:
        \begin{enumerate}
          \item if $f=g$, since $E$ is linear, for any integer $j\in\{1,\ldots,n\}$ it holds that $f\notin \Sigma_{E_j}$. 
            Then:
            
            \centerline{
              \begin{tabular}{l@{\ }l}
                $h\in \mathrm{Follow}(E,f,k)$ & $\Leftrightarrow$ $h\in \mathrm{First}(E_k)$\\   
                & $\Leftrightarrow$ $\exists t=h(t_1,\ldots,t_n)\in \llbracket E_k \rrbracket$ \ \ \ \  \textbf{(Lemma~\ref{firstComput})}\\   
                & $\Leftrightarrow$ $\exists t'=f(u_1,\ldots,u_{k-1},h(t_1,\ldots,t_n),\ldots,u_n) \in \llbracket E \rrbracket$\\       
                & $\Leftrightarrow$ $ h\in \phi^{f,k}_E$\\       
              \end{tabular}
            }
            
          \item Consider that there exists an integer $l$ such that $f\in\Sigma_{E_l}$. Then:
        
            \centerline{
              \begin{tabular}{l@{\ }l}
                & $h\in \mathrm{Follow}(E,f,k)$ \\
                $\Leftrightarrow$ & $h\in \mathrm{Follow}(E_l,f,k)$\\
                $\Leftrightarrow$ & $h\in \phi^{f,k}_{E_l}$ \ \ \ \ \textbf{(Induction hypothesis)}\\
                $\Leftrightarrow$ & $\exists t\in \llbracket E_l \rrbracket, \exists s\preccurlyeq t, \mathrm{root}(s)=f\wedge k\mbox{-}\mathrm{child(s)}=h$\\
                $\Leftrightarrow$ & $\exists t'=g(u_1,\ldots,u_{l-1},t,\ldots,u_n) \in \llbracket E \rrbracket, \exists s\preccurlyeq t, \mathrm{root}(s)=f\wedge k\mbox{-}\mathrm{child(s)}=h$\\
                $\Leftrightarrow$ & $\exists t'\in \llbracket E \rrbracket, \exists s\preccurlyeq t, \mathrm{root}(s)=f\wedge k\mbox{-}\mathrm{child(s)}=h$\\      
                $\Leftrightarrow$ & $ h\in \phi^{f,k}_E$\\       
              \end{tabular}
            }
        
          \item Otherwise, since $f\notin \Sigma_E$, $\mathrm{Follow}(E,f,k)=\emptyset=\phi^{f,k}_E$.
        \end{enumerate}
        
      \item Suppose that $E=E_1+E_2$. 
        Then:
        
        \centerline{
          \begin{tabular}{l@{\ }l@{\ }l}
            & $h\in \mathrm{Follow}(E,f,k)$\\
            $\Leftrightarrow$ & $h\in \mathrm{Follow}(E_1,f,k)\vee h\in \mathrm{Follow}(E_2,f,k)$\\
            $\Leftrightarrow$ & $h\in \phi^{f,k}_{E_1}\vee h\in \phi^{f,k}_{E_2}$ \ \ \ \ \textbf{(Induction hypothesis)}\\
            $\Leftrightarrow$ & $\exists t\in \llbracket E_1 \rrbracket, \exists s\preccurlyeq t, \mathrm{root}(s)=f\wedge k\mbox{-}\mathrm{child(s)}=h$\\
            & $\vee \exists t\in \llbracket E_2 \rrbracket, \exists s\preccurlyeq t, \mathrm{root}(s)=f\wedge k\mbox{-}\mathrm{child(s)}=h$\\
            $\Leftrightarrow$ & $\exists t\in \llbracket E_1 \rrbracket\cup \llbracket E_2 \rrbracket, \exists s\preccurlyeq t, \mathrm{root}(s)=f\wedge k\mbox{-}\mathrm{child(s)}=h$\\
            $\Leftrightarrow$ & $ h\in \phi^{f,k}_E$\\
          \end{tabular}    
        }
      
      \item Let us consider that $E=E_1\cdot_c E_2$. 
        The following cases have to be considered.
        \begin{enumerate}
          \item Let us suppose that $c$ is in $\mathrm{Follow}(E_1,f,k)$.
            \begin{enumerate} 
              \item\label{cas followcomput} Suppose that $h$ is in $\mathrm{Follow}(E_1,f,k)\setminus\{c\}$. 
                According to induction hypothesis, $h$ belongs to $\phi^{f,k}_{E_1}$ and then there exists a tree $t$ in $\llbracket E_1 \rrbracket$ that contains a subtree $s$ such that $\mathrm{root}(s)=f$ and $k\mbox{-}\mathrm{child(s)}=h$. 
                Since $h\neq c$, any tree $t'$ in $t_{c\leftarrow \llbracket E_2 \rrbracket}$ still contains the subtree $s$. Consequently, there exists a tree $t'$ in $\llbracket E \rrbracket$ that contains a subtree $s$ such that $\mathrm{root}(s)=f$ and $k\mbox{-}\mathrm{child(s)}=h$, \emph{i.e.} $h$ is in $\phi^{f,k}_E$.
              \item Suppose now that $h$ is in $\mathrm{First}(E_2)$. 
                According to Lemma~\ref{firstComput}, there exists a tree $t'=h(t_1,\ldots,t_n)$ in $\llbracket E_2 \rrbracket$. 
                Since $c$ is in $\mathrm{Follow}(E_1,f,k)$, then by induction hypothesis $c$ belongs to $\phi^{f,k}_{E_1}$ ; therefore there exists a tree $t_c$ in $\llbracket E_1 \rrbracket$ containing a subtree $s$ such that $\mathrm{root}(s)=f$ and $k\mbox{-}\mathrm{child(s)}=c$. 
                Consequently, $(t_c)_{c\leftarrow t'}$ admits by construction a subtree $s$ such that $\mathrm{root}(s)=f$ and $k\mbox{-}\mathrm{child(s)}=c$. 
                Therefore $h$ is in $\phi^{f,k}_E$.
              \item Conversely, suppose that $h\in \phi^{f,k}_E$. 
                Then there exists a tree $t'$ in $\llbracket E \rrbracket$ that contains a subtree $s$ such that $\mathrm{root}(s)=f$ and $k\mbox{-}\mathrm{child(s)}=h$.  
                By definition of $\llbracket E \rrbracket$, $t'$ is in $t_{c\leftarrow \llbracket E_2 \rrbracket}$ with $t$ a tree in $\llbracket E_1 \rrbracket$. 
                If $h$ is in $\Sigma_{E_1}\setminus\{c\}$, since $E$ is linear, $h$ appears in $t$ and so do $f$ and $s$. 
                Consequently, $h$ is in $\mathrm{Follow}(E_1,f,k)\subset \mathrm{Follow}(E,f,k)$.
                If $h$ is in $\Sigma_{E_2}$, then the subtree $s$ is created by substituting to a leave $c$ whose father is $f$ a tree in $\llbracket E_2 \rrbracket$ the root of which is $h$.
                Consequently $h$ is in $\mathrm{First}(E_2)\subset \mathrm{Follow}(E,f,k)$.
            \end{enumerate}
            
          \item Suppose now that $f\in\Sigma_{E_1}$ and that $c\notin \mathrm{Follow}(E_1,f,k)$. 
            \begin{enumerate} 
              \item Consider that $h\in \mathrm{Follow}(E_1,f,k)\setminus\{c\}$. 
                According to case~\ref{cas followcomput}, $h\in \phi^{f,k}_E$. 
              \item Conversely, suppose that $ h\in \phi^{f,k}_E$. 
                Then $\exists t'\in \llbracket E \rrbracket, \exists s\preccurlyeq t, \mathrm{root}(s)=f\wedge k\mbox{-}\mathrm{child(s)}=h$.  
                By definition, $t'\in t_{c\leftarrow \llbracket E_2 \rrbracket}$ with $t\in\llbracket E_1 \rrbracket$. 
                Since $c\notin \mathrm{Follow}(E_1,f,k)$, $h\in\Sigma_{E_1}\setminus\{c\}$ and since $E$ is linear, $h$ appears in $t$ and so do $f$ and $s$. 
                Consequently, $h\in \mathrm{Follow}(E_1,f,k)=\mathrm{Follow}(E,f,k)$.
            \end{enumerate}
            
          \item Consider that $f\in\Sigma_{E_2}$ and that $c\in  \mathrm{Last}(E_1)$.
            By definition, there exists a tree $t_c \in \llbracket E_1 \rrbracket$ such that $c\in\mathrm{Leaves}(t_c)$.
            \begin{enumerate} 
              \item Suppose that $h\in \mathrm{Follow}(E_2,f,k)$. 
                By induction hypothesis, $\exists t\in \llbracket E_2 \rrbracket, \exists s\preccurlyeq t, \mathrm{root}(s)=f, k\mbox{-}\mathrm{child(s)}=h$. 
                Moreover, $s$ is a subtree of $(t_c)_{c\Leftarrow t}\subset (t_c)_{c\Leftarrow \llbracket E_2 \rrbracket}\subset \llbracket E_ \rrbracket$. 
                Consequently, $\exists t\in \llbracket E \rrbracket, \exists s\preccurlyeq t, \mathrm{root}(s)=f, k\mbox{-}\mathrm{child(s)}=h$. 
                Hence, $h\in \phi^{f,k}_E$. 
              \item Conversely, suppose that $ h\in \phi^{f,k}_E$. 
                Then $\exists t'\in \llbracket E \rrbracket, \exists s\preccurlyeq t, \mathrm{root}(s)=f\wedge k\mbox{-}\mathrm{child(s)}=h$. 
                By definition, $t'\in t_{c\leftarrow \llbracket E_2 \rrbracket}$ with $t\in\llbracket E_1 \rrbracket$. 
                Since $f\in\Sigma_{E_2}$, so does $h$. 
                Consequently, $s$ is a subset that belongs to a tree $t''$ that has been substituted to a leave $c$ of $t$ to appear in $t'$. 
                Hence $h\in \mathrm{Follow}(E_2,f,k)=\mathrm{Follow}(E,f,k)$.
            \end{enumerate}
            
          \item In all other cases, there is no tree $t$ in $\llbracket E \rrbracket$ such that $\exists s\preccurlyeq t, \mathrm{root}(s)=f \wedge k\mbox{-}\mathrm{child(s)}=h$. 
            Hence $\mathrm{Follow}(E,f,k)=\emptyset=\phi^{f,k}_E$.
        \end{enumerate}
        
      \item Let us consider that $E=E_1^{*_c}$.
        \begin{enumerate}
          \item Consider that $h\in\mathrm{Follow}(E_1,f,k)$.
            By induction hypothesis, $\exists t\in \llbracket E_1 \rrbracket, \exists s\preccurlyeq t, \mathrm{root}(s)=f, k\mbox{-}\mathrm{child(s)}=h$.
            Since $\llbracket E_1 \rrbracket\subset \llbracket E \rrbracket$, $\exists t\in \llbracket E \rrbracket, \exists s\preccurlyeq t, \mathrm{root}(s)=f, k\mbox{-}\mathrm{child(s)}=h$.
            Hence $h\in \phi^{f,k}_E$.            
          \item Suppose that $h\in\mathrm{First}(E_1)$.
            According to Lemma~\ref{firstComput}, there exists a tree $t'=h(t_1,\ldots,t_n)$ in $\llbracket E_1 \rrbracket$.
            By construction, $c\in\mathrm{Follow}(E_1,f,k)$.
            By induction hypothesis, $c$ belongs to $\phi^{f,k}_{E_1}$ and then there exists a tree $t_c$ in $\llbracket E_1 \rrbracket$ that contains a subtree $s$ such that $\mathrm{root}(s)=f$ and $k\mbox{-}\mathrm{child(s)}=c$.
            Moreover, $t=(t_c)_{c\leftarrow t'}$ belongs by construction to $\llbracket E_1 \rrbracket\cdot_c \llbracket E_1 \rrbracket\subset \llbracket E_1 \rrbracket^{*_c}$ and contains a subtree $s$ such that $\mathrm{root}(s)=f$ and $k\mbox{-}\mathrm{child(s)}=h$.
            Therefore $\exists t\in \llbracket E \rrbracket, \exists s\preccurlyeq t, \mathrm{root}(s)=f, k\mbox{-}\mathrm{child(s)}=h$.
            Hence $h\in \phi^{f,k}_E$.            
          \item Conversely, suppose that $ h\in \phi^{f,k}_E$.  
            Then $\exists t\in \llbracket E \rrbracket, \exists s\preccurlyeq t, \mathrm{root}(s)=f, k\mbox{-}\mathrm{child(s)}=h$.
            By construction, $t\in t'_{c\leftarrow \llbracket E_1 \rrbracket}$ with $t'$ a tree in $\llbracket E_1 \rrbracket^{n_c}$ for some integer $n$.
            Let us consider that $t$ is the tree satisfying the previous condition with minimal $n$.
            If $s$ appears in $t'$, contradiction with the minimality of $n$.
            By construction, $t$ is obtained from $t'$ by substituting to its leaves $c$ a tree $u$ in $\llbracket E_1 \rrbracket$.
            If $s$ is a subtree of $u$, $h\in \phi^{f,k}_{E_1}$ and by induction hypothesis, $h\in \mathrm{Follow}(E_1,f,k)\subset \mathrm{Follow}(E,f,k)$.
            If $s$ does not appear in $u$, then it is created during the $c$-product, \emph{i.e.} $t'$ admits a leave $c$ the root of which is $f$, and $u$ admits $h$ as root.
            Consequently, by induction hypothesis, $c\in\mathrm{Follow}(E_1,f,k)$ and according to Lemma~\ref{firstComput}, $h\in\mathrm{First}(E_1)\subset \mathrm{Follow}(E,f,k)$.
        \end{enumerate}
    \end{enumerate}
  \end{proof}
  
   The two functions $\mathrm{First}$ and $\mathrm{Follow}$ are sufficient to compute the \emph{$K$-position tree automaton} of $E$. 
 
  \begin{definition}\label{def aut pos}
    Let $\E$ be linear.
    The $K$-\emph{position automaton} ${\cal P_{\E}}$ is the automaton $(Q,\Sigma,Q_T,\Delta)$ defined by 
    
    \centerline{
      \begin{tabular}{l@{\ }l}
        $Q=$ & $\ \ \{f^k \mid f\in \Sigma_m\wedge 1\leq k\leq m\}$\\
        & $\cup \{\varepsilon^1\}$ with $\varepsilon^1$ a new symbol not in $\Sigma$, $Q_T=\{\varepsilon^1\}$,\\
      \end{tabular}}    
    
      \centerline{$\begin{array}{r@{\ }c@{\ }l}
      \Delta = & & \{(f^k,g,g^1,\ldots,g^n)\mid f\in\Sigma_m\wedge k\leq m\wedge g\in\Sigma_n\wedge g \in \Follow(\E,f,k)\} \\
            & \cup &  \{(\varepsilon^1,f,f^1,\ldots,f^m)\mid f \in \Sigma_m\wedge f \in\First(\E)\}\\
            &  \cup &  \{(\varepsilon^1,c)\mid c\in \Sigma_0\wedge c\in\First(\E)\}\\
            & \cup & \{(f^k,c)\mid f\in\Sigma_m\wedge k\leq m\wedge c\in \Follow(\E,f,k)\}\\
      \end{array}$}
      
  \end{definition}
  
  In order to show that the $K$-position tree automaton of $\E$ accepts $\llbracket \E\rrbracket$, we characterize the membership of a tree $t$ in 
   the language denoted by $\E$ using the functions $\mathrm{First}$ and $\mathrm{Follow}$.
  
   \begin{proposition}\label{proposition Le}
   Let $\E$ be linear. A tree $t$ belongs to $\llbracket \E\rrbracket$ if and only if:
    \begin{enumerate}
      \item $\mathrm{root}(t)\in\First(\E)$ and
      \item for every subtree $f(t_1,\ldots,t_m)$ of $t$, for any integer $k$ in $\{1,\ldots,m\}$, $\mathrm{root}(t_k)\in\Follow(\E,f,k)$.
    \end{enumerate}
  \end{proposition}
  \begin{proof}
    If $t$ belongs to $\llbracket \E\rrbracket$, it holds by definition that $\rooot(t)\in \First(\E)$ and that $\forall f(t_1,\ldots,t_m)\preccurlyeq t, \forall 1\leq k\leq m, \rooot(t_k)\in \Follow(\E,f,k)$.
    Let us show the reciprocal part by induction over the structure of $E$, \emph{i.e.} \textbf{(I)} $\rooot(t)\in \First(\E)$ and \textbf{(II)}$\forall f(t_1,\ldots,t_m)\preccurlyeq t, \forall 1\leq k\leq m, \rooot(t_k)\in \Follow(\E,f,k)$ $\Rightarrow$ $t\in\llbracket \E\rrbracket$.    
    \begin{enumerate}
      \item If $E=a$, $\mathrm{First}(E)=\{a\}$. 
        Hence $\rooot(t)\in \First(\E) \Rightarrow t=a \Rightarrow t\in \llbracket \E\rrbracket$.
        
      \item Let us consider that $E=g(E_1,\ldots,E_n)$. 
        In this case, $\First(\E)=\{g\}$.  
        From \textbf{(I)}, $\mathrm{root}(t)=g\in \First(\E)$,  and then $t=g(u_1,\ldots,u_n)$.
        Let $u_j\in\{u_1,\ldots,u_n\}$.
        From \textbf{(II)} $\rooot(u_j)\in \Follow(\E,f,j)=\mathrm{First}(E_j)$ \textbf{(A)}.
        Since $u_j$ is a subtree of $t$, Condition \textbf{(II)} also implies that $\forall f(t_1,\ldots,t_m)\preccurlyeq u_j, \forall 1\leq k\leq m, \rooot(t_k)\in \Follow(\E,f,k)$ \textbf{(B)}.
        Since $f$ appears in $u_j$, $f\in\Sigma_{E_j}$ and then $\Follow(\E,f,k)=\Follow(\E_j,f,k)$ \textbf{(C)}.
        By induction hypothesis, the conditions \textbf{(A)}, \textbf{(B)} and \textbf{(C)} implies that $u_j\in \llbracket \E_j\rrbracket$.
        Consequently $t=g(u_1,\ldots,u_n)\in \llbracket \E\rrbracket$.
        
      \item Let us suppose that $E=E_1+E_2$. Then:
      
        \centerline{
          \begin{tabular}{l@{\ }l}
            &  $\rooot(t)\in \First(\E),$\\
            &  $\forall f(t_1,\ldots,t_m)\preccurlyeq t, \forall 1\leq k\leq m, \rooot(t_k)\in \Follow(\E,f,k)$\\
            $\Rightarrow$ &  $\rooot(t)\in \First(\E_j),$\\
            & $\forall f(t_1,\ldots,t_m)\preccurlyeq t, \forall 1\leq k\leq m, \rooot(t_k)\in \Follow(\E_j,f,k), j\in\{1,2\}$\\
            $\Rightarrow$&  $t\in \llbracket \E_j\rrbracket, j\in\{1,2\}$\ \ \ \  \textbf{(Induction hypothesis)}\\
            $\Rightarrow$ &  $t\in \llbracket \E_1\rrbracket\cup \llbracket E_2\rrbracket$\\
          \end{tabular}
        }
        
      \item Let us consider that $E=E_1\cdot_c E_2$.
        Two cases have to be considered.
        \begin{enumerate}
          \item If $\rooot(t)$ is in $\First(\E_2)$, then $c$ is in $\llbracket \E_1\rrbracket$. 
            Consequently, any symbol appearing in $t$ belongs to $\Sigma_{E_2}$ by definition of $\mathrm{Follow}$ and following \textbf{(II)}. 
            Hence, for any subtree $f(t_1,\ldots,t_m)$ of $t$, for any integer $k$ in $\{1,\ldots,m\}$, the root of $t_k$ is in $\Follow(\E_2,f,k)$.
            By induction hypothesis, $t\in \llbracket \E_2\rrbracket$.
            Furthermore $\llbracket \E_2\rrbracket=c\cdot_c \llbracket \E_2\rrbracket\subset \llbracket \E_1\rrbracket\cdot_c \llbracket \E_2\rrbracket$.
            Consequently $t$ is in $\llbracket \E_1\rrbracket\cdot_c \llbracket \E_2\rrbracket$.
          \item Suppose that $\rooot(t)$ is in $\First(\E_1)$. 
            Necessarily, by definition of $\mathrm{First}$, the root of $t$ cannot be $c$.
            
            \noindent Let us construct the tree $u=\phi(t)$ where $\phi$ is inductively defined for any tree $t'$ as follows:
            
            \centerline{
              $\phi(t')=
                \left\{
                  \begin{array}{l@{\ }l}
                    f(\phi(t_1),\ldots,\phi(t_n)) & \text{ if } t'=f(t_1,\ldots,t_n)\\
                    & \ \ \wedge (f=\rooot(t)\vee f\in \Sigma_{E_1}),\\
                    c & \text{otherwise}.
                  \end{array}
                \right.$ 
            }
            
            \noindent Notice that a subtree $s$ of $t$ is substituted by $c$ when its root is in $\mathrm{First}(E_2)$.
            Therefore, if $s$ is the $k$-th child of a node $f$, then it holds $c\in\mathrm{Follow}(E_1,f,k)$.
            
            \noindent Since by construction, \textbf{(I)} implies $\rooot(u)=\rooot(t)$ is in $\First(\E_1)$ and \textbf{(II)} implies $\forall f(t_1,\ldots,t_m)\preccurlyeq u, \forall 1\leq k\leq m$, $\rooot(t_k)$ is in $\Follow(\E_1,f,k)$, it holds according to induction hypothesis that $u\in \llbracket \E_1\rrbracket$.
            Furthermore, let us consider the set $V=\{v\mid g(r_1,\ldots,v,\ldots,r_l) \preccurlyeq t \wedge \rooot(v)\in \Sigma_{E_2}\wedge g\in \Sigma_{E_1} \}$.
            Let $v$ be a tree in $V$.
            By construction, $v$ satisfies $\rooot(v)\in\First(\E_2)$ since it is in $\mathrm{Follow}(E,g,x)$ for some integer $x$.
            Moreover, \textbf{(II)} implies that $\forall f(t_1,\ldots,t_m)\preccurlyeq v, \forall 1\leq k\leq m$, $\rooot(t_k)$ is in $\Follow(\E,f,k)=\Follow(\E_2,f,k)$. 
            It holds according to induction hypothesis that $v\in \llbracket \E_2\rrbracket$.
            Finally, $t$ can be obtained from $u$ by substituting to the leaves labelled by $c$ a tree in $V$.
            Therefore $t$ is in $t'_{c\leftarrow V}\subset t'_{c\leftarrow \llbracket \E_2\rrbracket}\subset \llbracket \E_1\rrbracket\cdot_c \llbracket \E_2\rrbracket$.
        \end{enumerate}
        
      \item Let us suppose that $E=E_1^{*_c}$.
        If $\rooot(t)=c$ then $t=c$ and $t\in\llbracket \E\rrbracket $. 
        Otherwise, $\rooot(t)$ is in $\First(\E_1)$. 
        Furthermore, \textbf{(II)} implies that $\forall f(t_1,\ldots,t_m)\preccurlyeq t$, $\forall 1\leq k\leq m$, $\rooot(t_k)$ is either in $\mathrm{First}(E_1)$ if $c\in\mathrm{Follow}(E_1,f,k)$, or in $\Follow(\E_1,f,k)$.
        
        \noindent Let us construct the tree $u=\phi(t)$ where $\phi$ is inductively defined for any tree $t'$ as follows:
        
        \centerline{
          $\phi(t')=
            \left\{
              \begin{array}{l@{\ }l}
                f(\phi(t_1),\ldots,\phi(t_n)) & \text{ if } t'=f(t_1,\ldots,t_n)\\
                & \ \ \wedge (f=\rooot(t)\vee f\notin \mathrm{First}(E_1)),\\
                c & \text{otherwise}. 
              \end{array}
            \right.$      
        }
        
        \noindent Notice that a subtree $s$ is substituted by $c$ when its root is in $\mathrm{First}(E_1)$.
        Therefore, if $s$ is the $k$-th child of a node $f$, then it holds $c\in\mathrm{Follow}(E_1,f,k)$.
        
        Since by construction, \textbf{(I)} implies that $\rooot(u)=\rooot(t)$ is in $\First(\E_1)$, and since \textbf{(II)} implies that $\forall f(t_1,\ldots,t_m)\preccurlyeq u, \forall 1\leq k\leq m$, $\rooot(t_k)$ is in $\Follow(\E_1,f,k)$, it holds according to induction hypothesis that $u\in \llbracket \E_1\rrbracket$.
        Furthermore, let us consider the set $V=\{v\mid g(r_1,\ldots,v,\ldots,r_l) \preccurlyeq t \wedge \rooot(v)\in \mathrm{First}(E_1) \}\setminus\{t\}$.
        Let $v$ be a tree in $V$.
        By construction, $v$ satisfies $\rooot(v)\in\First(\E_1)$, and $\forall f(t_1,\ldots,t_m)\preccurlyeq v, \forall 1\leq k\leq m$,  either $\rooot(t_k)$ is in $\mathrm{First}(E_1)$ if $c\in\mathrm{Follow}(E_1,f,k)$, or in $\Follow(\E_1,f,k)$.
        
        \noindent The function $\phi$ can be applied over any tree $v$ in $V$ to produce another tree in $\llbracket \E_1\rrbracket$ and another set of smaller trees $V'$.
        By recurrence, it can be shown that this process halts and that any tree in $V$ belongs to $\llbracket \E_1\rrbracket^{n_c}$ for some integer $n$.
        As a direct consequence, since $t$ can be obtained from $u$ by substituting to the leaves labelled by $c$ a tree in $V$, the tree $t$ is in $u_{c\leftarrow V}\subset u_{c\leftarrow \llbracket \E_1\rrbracket^{n_c}}\subset \llbracket \E_1\rrbracket^{(n+1)_c}\subset \llbracket \E_1\rrbracket^{*_c}$.
    \end{enumerate}
  \end{proof}
  
  Let us show how to link the characterization in Proposition~\ref{proposition Le} with the transition sequences in ${\cal P_{\E}}$.
  
  \begin{proposition}\label{proposition Lpe}
    Let $\E$ be linear and ${\cal P_{\E}}=(Q,\Sigma,Q_T,\Delta)$.
    Let $t=f(t_1,\ldots,t_m)$ be a term in $T_{\Sigma}$.
    Then the two following propositions are equivalent:
    \begin{enumerate}
      \item $\forall g(s_1,\ldots,s_l)\preccurlyeq t$, $\forall p\leq l$, $\rooot(s_p)\in \Follow(\E,g,p)$,
      \item $\forall 1\leq k\leq m$, $f^k\in \Delta^*(t_k)$.
    \end{enumerate}
  \end{proposition}
  \begin{proof}
    By induction over $t$.
    \begin{itemize}
      \item Let $t=f(a_1,\ldots,a_m)$ with $a_j\in\Sigma_0$ for any integer $1\leq j\leq m$. 
        Hence, by construction of ${\cal P}_{\E}$, for any integer $j\leq m$, $\rooot(a_j)\in \Follow(\E,f,j)$ $\Leftrightarrow$ $f^k\in\Delta(a_j)$.
      \item Suppose that $t=f(t_1,\ldots,t_m)$, with $t_j=g_j(r_{j,1},\ldots,r_{j,n_j})$ for any integer $1\leq j\leq m$.      
        Let $j$ be an integer in $\{1,\ldots,n\}$. 
        By induction hypothesis, the two following conditions are equivalent:
        \begin{enumerate}
          \item $\forall g(s_1,\ldots,s_l)\preccurlyeq t_j$, $\forall p\leq l$, $\rooot(s_p)\in \Follow(\E,g,p)$,
          \item $\forall 1\leq k\leq n_j$, $g^k_j\in \Delta^*(r_{j,k}$.
        \end{enumerate}
        
        \noindent Hence:
        
        \centerline{
          \begin{tabular}{l@{\ }l}
            & $\forall g(s_1,\ldots,s_l)\preccurlyeq t$, $\forall p\leq l$, $\rooot(s_p)\in \Follow(\E,g_i,p)$\\
            $\Leftrightarrow$ & $\forall 1\leq j\leq m$, $\forall g(s_1,\ldots,s_l)\preccurlyeq t_j$, $\forall p\leq l$, $\rooot(s_p)\in \Follow(\E,g,p)$\\
            & $\wedge$ $g_j=\rooot(t_j)\in \Follow(\E,f,j)$\\
            $\Leftrightarrow$ & $\forall 1\leq j\leq m$, $\forall 1\leq k\leq n_j$, $g^k_j\in \Delta^*(s_k)$ $\wedge$ $(f^j,g_j,g^1_j,\ldots,g^l_j)\in\Delta)$\\
            $\Leftrightarrow$ & $\forall 1\leq k\leq m$, $f^k\in \Delta^*(t_k)$.\\
          \end{tabular}
        }
    \end{itemize}
  \end{proof}
  
  As a direct consequence of the two previous propositions, it can be shown that the $K$-position automaton of $\E$ recognizes the language denoted by $E$.
  
  \begin{theorem}\label{thm lang pe eq e}
     If $\E$ is linear, then ${\cal L}({\cal P}_{\E})=\llbracket \E\rrbracket$.
  \end{theorem}
  \begin{proof}
    Let $t$ be a tree in $T_\Sigma$. 
    
    \noindent If $t=a\in\Sigma_0$, $a\in {\cal L}({\cal P}_{\E})$ $\Leftrightarrow$ $a\in\mathrm{First}(E)$ $\Leftrightarrow$ $a\in\llbracket \E\rrbracket$.

    \noindent If $t=f(t_1,\ldots,t_n)$:
    
    \centerline{
      \begin{tabular}{l@{\ }l}
        & $t\in {\cal L}({\cal P}_{\E})$\\
        $\Leftrightarrow$ & $\varepsilon^{1}\in \Delta^*(t)$\\
        $\Leftrightarrow$ & $(\varepsilon^{1},f,f^1,\ldots,f^n)\in\Delta$ $\wedge$ $\forall 1\leq j\leq n$, $f_j\in\Delta^{*}(t_j)$\\
        $\Leftrightarrow$ & $f\in\mathrm{First}(E)$\\
        & \ \ $\wedge$ $\forall g(s_1,\ldots,s_l)\preccurlyeq t$, $\forall p\leq l$, $\rooot(s_p)\in \Follow(\E,g,p)$ \ \ \ \ (Proposition~\ref{proposition Lpe})\\
        $\Leftrightarrow$ & $t\in \llbracket \E\rrbracket$ \ \ \ \ (Proposition~\ref{proposition Le})\\
      \end{tabular}
    }
  \end{proof}
  
  This construction can be extended 
to expressions that are not necessarily linear
  using the linearization and the mapping $h$.  
    The \emph{$K$-Position Automaton} ${\cal P_{\E}}$ associated with $\E$ is obtained by replacing each transition  $(f^k_j,g_i,g^1_{i},\dots,g^n_{i})$ of the tree automaton ${\cal P_{\b \E}}$ by $(f^k_j,h(g_i),g^1_{i},\dots,g^n_{i})$.

  \begin{corollary}
    
    %
      $h(\llbracket\b\E\rrbracket)=h({\cal L}({\cal  P_{\b\E}}))={\cal L}({\cal P_{\E}})=\llbracket\E\rrbracket$. 
  \end{corollary}

\begin{example}
Let $\E=(f(a)^{*_a}\cdot_a b+ h(b))^{*_b}+g(c,a)^{*_c}\cdot_c (f(a)^{*_a}\cdot_a b+ h(b))^{*_b}$ be the regular expression of 
Example~\ref{Pos Automat} and $\b\E=(f_1(a)^{*_a}\cdot_a b+ h_2(b))^{*_b}+g_3(c,a)^{*_c}\cdot_c (f_4(a)^{*_a}\cdot_a b+ h_5(b))^{*_b}$ its linearized form. The $k$-Position Automaton ${\cal  P}_{\b\E}$ associated with $\b\E$ is given in Figure~\ref{fig a t e2}.
%
The set of states is $Q=\{\varepsilon^1,f^1_1,h^1_2,g^1_3,g^2_3,f^1_4,h^1_5\}$. The set of final states is $Q_T=\{\varepsilon^1\}$. 
The set of transition rules $\Delta$ is 

\centerline{
    $f_1(f_1^1)\rightarrow \varepsilon^1$, $f_1(f_1^1)\rightarrow f_1^1$, $f_1(f_1^1)\rightarrow h_2^1$,
}

\centerline{
    $h_2(h_2^1)\rightarrow \varepsilon^1$, $h_2(h_2^1)\rightarrow f_1^1$, $h_2(h_2^1)\rightarrow h_2^1$,
}

\centerline{
    $g_3(g_3^1,g_3^2)\rightarrow g_3^1$, $g_3(g_3^1,g_3^2)\rightarrow \varepsilon^1$,
}

\centerline{
    $f_4(f_4^1)\rightarrow \varepsilon^1$, $f_4(f_4^1)\rightarrow g_3^1$, $f_4(f_4^1)\rightarrow f_4^1$, $f_4(f_4^1)\rightarrow h_5^1$,
}

\centerline{
    $h_5(h_5^1)\rightarrow \varepsilon^1$, $h_5(h_5^1)\rightarrow g_3^1$, $h_5(h_5^1)\rightarrow f_4^1$, $h_5(h_5^1)\rightarrow h_5^1$,
}

\centerline{
    $a\rightarrow g_3^2$,
}

\centerline{
    $b\rightarrow \varepsilon^1$, $b\rightarrow f_1^1$, $b\rightarrow h_2^1$, $b\rightarrow g_3^1$, $b\rightarrow f_4^1$, $b\rightarrow h_5^1$
}

\noindent The number of states is $|Q|=7$ and the number of transition rules is $|\Delta|=23$. 

\end{example}
\begin{figure}[H]
  \centerline{
	\begin{tikzpicture}[node distance=2.5cm,bend angle=30,transform shape,scale=1]
	  \node[accepting,state] (eps) {$\varepsilon^1$};
	  \node[state, above left of=eps] (f11) {$f^1_1$};	
	  \node[state, above right of=eps] (h12) {$h^1_2$}; 
      \node[state, below of=eps] (g13) {$g^1_3$};
      \node[state, right of=g13,node distance=3.5cm] (g23) {$g^2_3$};
	  \node[state, below left of=g13,node distance=3.5cm] (h15) {$h^1_5$};
      \node[state, below right of=g13,node distance=3.5cm] (f14) {$f^1_4$};
	  \draw (eps) ++(-1cm,0cm) node {$b$}  edge[->] (eps);  
	  \draw (f11) ++(-1cm,0cm) node {$b$}  edge[->] (f11);  
	  \draw (h12) ++(1cm,0cm) node {$b$}  edge[->] (h12); 
	  \draw (h15) ++(0cm,-1cm) node {$b$}  edge[->] (h15);  
	  \draw (g23) ++(1cm,0cm) node {$a$}  edge[->] (g23);  
	  \draw (g13) ++(-1cm,0cm) node {$b$}  edge[->] (g13);    
	  \draw (f14) ++(0cm,-1cm) node {$b$}  edge[->] (f14);
      \path[->]
        (f11) edge[->,below left] node {$f_1$} (eps)
		(f11) edge[->,loop,above] node {$f_1$} ()
		(h12) edge[->,bend right,above] node {$h_2$} (f11)
		%
		(h12) edge[->,loop,above] node {$h_2$} ()
		(h12) edge[->,below right] node {$h_2$} (eps)
		(f11) edge[->,bend right,above] node {$f_1$} (h12)
	    %
		(h15) edge[->, in=135,out=-135,loop,left] node {$h_5$} ()	
		(h15) edge[->,above left] node {$h_5$} (eps)		
		(h15) edge[->,above left] node {$h_5$} (g13)		
		(h15) edge[->,bend right,above] node {$h_5$} (f14)	  
		(f14) edge[->,in=45,out=-45,loop,right] node {$f_4$} ()	
		(f14) edge[->,bend right,above] node {$f_4$} (h15)		
		(f14) edge[->,above right] node {$f_4$} (eps)		
		(f14) edge[->,above right] node {$f_4$} (g13)
	  ;
      \draw (eps) ++(1.75cm,-0.75cm)  edge[->,in=0,out=90] node[above right,pos=0] {$g_3$} (eps) edge (g13) edge (g23);  
      \draw (g13) ++(1.5cm,1cm)  edge[->,in=90,out=145] node[above right,pos=0] {$g_3$} (g13) edge (g13) edge (g23);
    \end{tikzpicture}
  }
  \caption{The $k$-Position Automaton ${\cal P}_{\b\E}$.}
  \label{fig a t e2}
\end{figure}	

  \subsection{The Follow Tree Automaton}\label{subsec folltreeaut}
  
  In this section, we define the follow tree automaton which is a generalisation of the Follow automaton introduced by L.~Ilie and S.~Yu in~\cite{yu} in the case of words, and that it is a quotient of the $K$-position automaton, similarly to the case of 
  words. 
  Notice that in this automaton, states are no longer positions, but 
  sets
   of positions and that we extend the definition of the function $\Follow$ to the position $\varepsilon^1$ by $\mathrm{Follow}(E,\varepsilon^1,1)=\mathrm{First}(E)$. 
  
  \begin{definition}
    Let $\E$ be linear. 
    The \emph{Follow Automaton} of $\E$ is the tree automaton ${\cal F}_{\E}=(Q,\Sigma,Q_T,\Delta)$ defined as follows: 
    
      
      
      \centerline{$
      \begin{array}{r@{\ }c@{\ }l}
      Q=& &\{ \First(\E)\}\cup \bigcup_{f\in{{\Sigma_{E}}_m}} \{\Follow(\E,f,k) \mid 1\leq k\leq m \}\\
      Q_T=& &\{\First(\E)\}\\
      \Delta= &   &\lbrace(\Follow(\E,g,l),f,\Follow(\E,f,1) , \ldots, \Follow(\E,f,m) \mid f\in {{\Sigma_{E}}_m}\\
      & & \ \ \ \ \wedge f\in \Follow(\E,g,l)\wedge  g\in\Sigma_n\wedge l\leq n \}\\
              & & \cup\{ (I,c)\mid c\in I \wedge c\in\Sigma_0\rbrace\\
      \end{array}
              $}
  \end{definition}
  
  Let us show that ${\cal F_{\E}}$ is a quotient of ${\cal P_{\E}}$ w.r.t. a similarity relation ; since this kind of quotient preserves the language, this method is consequently a 
  proof
   of the fact that the language denoted by $E$ is recognized by ${\cal F}_{\E}$.
  
   A \emph{similarity relation} over an automaton 
   $A=(Q,\Sigma,Q_T,\Delta)$
    is 
    an equivalence relation
    $\sim$ over $Q$ such that for any two states $q$ and $q'$ in $Q$: $q\sim q'$ $\Rightarrow$ $\forall f\in \Sigma_n$, $\forall (q_1,\dots,q_n)\in Q^n$, $(q,f,q_1,\dots,q_n)\in \Delta$ $\Leftrightarrow$ $(q',f,q_1,\dots,q_n)\in \Delta$.
    In other words, two similar states admit the same predecessors w.r.t. any symbol.

	\begin{proposition}\label{similar relation}
	  Let $\cal A$ be an automaton and $\sim$ be a similarity relation over $\cal A$. 
	  Then $ {\cal L}({\cal A}_{/\sim})={\cal L}({\cal A})$.
	\end{proposition}
	\begin{proof}  
	  Let us set ${\cal A}=(Q,\Sigma,Q_T,\Delta)$ and ${\cal A}_{/\sim}=(Q',\Sigma,Q'_T,\Delta')$.
    Let $t$ be tree in $T_\Sigma$.
    Let us show by induction over the structure of $t$ that $q\in\Delta^*(t) \Leftrightarrow [q]\in\Delta'^*(t)$.
    \begin{enumerate}
      \item Consider that $t=a\in\Sigma_0$.
        Then By construction of ${\cal A}_{/\sim}$, $(q,a)\in\Delta$ $\Leftrightarrow$ $([q],a)\in\Delta'$.  
      \item Let us consider that $t=f(t_1,\ldots,t_n)$ with $f$ in $\Sigma_n$ and $t_1,\ldots,t_n$ any $n$ trees in $T_\Sigma$.
        Then:
        
        \centerline{
          \begin{tabular}{l@{\ }l}
            & $q\in\Delta^*(t)$\\ 
            $\Leftrightarrow$ &  $q\in\Delta^*(f(t_1,\ldots,t_n))$\\
            $\Leftrightarrow$ &  $q\in \Delta(f,q_1,\ldots,q_n)$ with $q_j\in\Delta^*(q_j)$ and $1\leq j\leq n$\\
            $\Leftrightarrow$&  $q\in \Delta(f,q_1,\ldots,q_n)$ with $[q_j]\in\Delta'^*([q_j])$ and $1\leq j\leq n$\\
             & \textbf{(Induction hypothesis)}\\
            $\Leftrightarrow$&  $q\in \Delta'(f,[q_1],\ldots,[q_n])$ with $[q_j]\in\Delta'^*([q_j])$ and $1\leq j\leq n$\\
            & \textbf{(Construction of ${\cal A}_{/\sim}$)}\\
            $\Leftrightarrow $&  $[q]\in\Delta'^*(f(t_1,\ldots,t_n))$\\
            $\Leftrightarrow$ &  $[q]\in\Delta'^*(t)$\\
          \end{tabular}
        }
    \end{enumerate}
	\end{proof}
	
	The quotient from $\mathcal{P}_E$ to $\mathcal{F}_E$ is defined by the following similarity relation. 	
	  Let $\E$ be linear and ${\cal P}_{\E}=(Q,\Sigma,Q_T,\Delta)$.
	  The \emph{Follow Relation} is the relation $\sim_{\cal F}$ defined for any two states $f^k$ and $g^l$ in $Q$ by $f^k \sim_{\cal F} g^l$ $\Leftrightarrow$ $\Follow(\E,f,k)=\Follow(\E,g,l)$.

	\begin{proposition}\label{prop simF largest sim}
	  Let $E$ be linear.
	  The relation $\sim_{\cal F}$ is the largest similarity relation over $\mathcal{P}_E$.
	\end{proposition}
	\begin{proof}  
	  Let us set ${\cal P}_E=(Q,\Sigma,Q_T,\Delta)$.
	  Let $f^k$ and $g^l$ be two states in $Q$.
	  Then:
    
    \centerline{
      \begin{tabular}{l@{\ }l}
        & $f^k \sim_{\cal F} g^l$\\ 
        $\Longleftrightarrow$ & $\Follow(\E,f,k)=\Follow(\E,g,l)$\\
        $\Longleftrightarrow$ & $\forall h \in\Sigma_m$, $(h\in \Follow(\E,f,k)\Leftrightarrow h\in \Follow(\E,g,l))$\\
        $\Longleftrightarrow$ & $\forall h \in\Sigma_m$, $((f^k,h,h^1,\ldots,h^m)\in\Delta \Leftrightarrow (g^l,h,h^1,\ldots,h^m)\in\Delta)$\\
        & \textbf{(Construction of ${\cal P}_E$)}\\
        $\Longleftrightarrow$ & $\forall h\in \Sigma_m$, $\forall (q_1,\dots,q_m)\in Q^m$, $(f^k,h,q_1,\dots,q_n)\in \Delta$ $\Leftrightarrow$ $(g^l,f,q_1,\dots,q_n)\in \Delta)$\\
        & \textbf{(Construction of ${\cal P}_E$)}\\
      \end{tabular}
    }
  \end{proof}
  
  \begin{proposition}\label{prop quot eq}
    Let $\E$ be linear.
    The finite tree automaton ${\cal P_{\E}}\diagup_{\sim_{\cal F}}$ is isomorphic to ${\cal F}_{\E}$.
  \end{proposition}
  \begin{proof} 
	  Let us set ${\cal P_{\E}}\diagup_{\sim_{\cal F}}=(Q,\Sigma,Q_T,\Delta)$ and ${\cal F}_{\E}=(Q',\Sigma,Q'_T,\Delta')$.
	  Let us consider the mapping $\phi$ defined from $Q$ to $Q'$ by $\phi([f^k])=\mathrm{Follow}(E,f,k)$ ($\phi([\varepsilon^1])=\mathrm{Follow}(E,\varepsilon^1,1)=First(E)$).
	  Let us show that $([f^k],g,[g^1],\ldots,[g^m])$ is a transition in $\Delta$ if and only if $(\phi([f^k]),g,\phi([g^1]),\ldots,\phi([g^m]))$ is a transition in $\Delta'$.
	  Let $f^k$ and $g^l$ be two states in $Q$.
	  Then:
	  
	  \centerline{
	    \begin{tabular}{l@{\ }l}
	      & $([f^k],g,[g^1],\ldots,[g^m])\in\Delta$\\
	      $\Leftrightarrow$ & $g\in\mathrm{Follow}(E,f,k)$\ \ \ \ \textbf{(Definition of $\sim_{\cal F}$)}\\
	      $\Leftrightarrow$ & $(\mathrm{Follow}(E,f,k),g,\mathrm{Follow}(E,g,1),\ldots,\mathrm{Follow}(E,g,m))\in\Delta$\\
	      & \textbf{(Definition of ${\cal F}_{\E}$)}\\
	      $\Leftrightarrow$ & $(\phi([f^k]),g,\phi([g^1]),\ldots,\phi([g^m]))\in\Delta'$\\
	    \end{tabular}
	  }
	\end{proof}
	
	As a direct consequence of the previous results, the following theorem can be shown.
	
	\begin{theorem}\label{thm LFE eq LE}
	  Let $\E$ be linear. 
	  Then ${\cal L}({\cal F}_{\E})=\llbracket \E\rrbracket$.
	\end{theorem}
	\begin{proof}
	  According to Theorem~\ref{thm lang pe eq e}, $\llbracket \E\rrbracket={\cal L}({\cal P}_{\E})$.
	  According to Proposition~\ref{prop quot eq}, ${\cal F}_{\E}$ is a quotient of ${\cal P}_{\E}$ w.r.t. $\sim_{\mathcal{F}}$, which is a similarity relation following Proposition~\ref{prop simF largest sim}.
	  Consequently, Proposition~\ref{similar relation} asserts that ${\cal L}({\cal P}_{\E})={\cal L}({\cal F}_{\E})$.
	  Therefore ${\cal L}({\cal F}_{\E})=\llbracket \E\rrbracket$.
	\end{proof}	
	
  \noindent Finally, this method can be extended to 
  expressions that are not necessarily linear  
  as follows.	
  The \emph{Follow Automaton} ${\cal F_{\E}}$ associated with $\E$ is obtained by replacing each transition  $(I,f_j,\Follow(\E,f_j,1) , \ldots,$ $\Follow(\E,f_j,m))$ of ${\cal F_{\b \E}}$ by $(I,h(f_j),\Follow(\E,f_j,1) , \ldots, \Follow(\E,f_j,m))$.
  
	\begin{corollary}
	${\cal L}({\cal F}_{\E})=\llbracket \E\rrbracket$.	
	\end{corollary}

\begin{example} \label{exp follow automaton}
The Follow Automaton ${\cal F}_{\E}$ associated with  the tree expression $\E=(f(a)^{*_a}\cdot_a b+ h(b))^{*_b}+g(c,a)^{*_c}\cdot_c (f(a)^{*_a}\cdot_a b+ h(b))^{*_b}$ of 
 Example~\ref{Pos Automat} is given in Figure~\ref{fig r t e3}.

The set of states is $Q=\{\{a\},\{b,f_1,h_2\},\{b,f_1,h_2,g_3,f_4,h_5\},\{b,g_3,f_4,h_5\},$ $\{b,f_4,h_5\}\}$ and $Q_T=\{\{b,f_1,h_2,g_3,f_4,h_5\}\}$. 
 The set of transition rules $\Delta$ is

%
\centerline{
  $f(\{b,f_1,h_2\})\rightarrow \{b,f_1,h_2,g_3,f_4,h_5\}$, $f(\{b,f_1,h_2\})\rightarrow \{b,f_1,h_2\}$
}

\centerline{
  $h(\{b,f_1,h_2\})\rightarrow \{b,f_1,h_2,g_3,f_4,h_5\}$,
  $h(\{b,f_1,h_2\})\rightarrow \{b,f_1,h_2\}$,
}

\centerline{
  $g(\{b,g_3,f_4,h_5\},g_3^2)\rightarrow \{b,g_3,f_4,h_5\}$,
}

\centerline{
  $g(\{b,g_3,f_4,h_5\},g_3^2)\rightarrow \{b,f_1,h_2,g_3,f_4,h_5\}$,
}

\centerline{
  $f(\{b,f_4,h_5\})\rightarrow \{b,f_1,h_2,g_3,f_4,h_5\}$,
}

\centerline{
  $f(\{b,f_4,h_5\})\rightarrow \{b,g_3,f_4,h_5\}$,
}

\centerline{
  $f(\{b,f_4,h_5\})\rightarrow \{b,f_4,h_5\}$,
}

\centerline{
  $h(\{b,f_4,h_5\})\rightarrow \{b,f_1,h_2,g_3,f_4,h_5\}$,
}
  
\centerline{
  $h(\{b,f_4,h_5\})\rightarrow \{b,g_3,f_4,h_5\}$,
}
  
\centerline{
  $h(\{b,f_4,h_5\})\rightarrow \{b,f_4,h_5\}$,
}

\centerline{$a\rightarrow \{a\}$,}

\centerline{
  $b\rightarrow \{b,f_1,h_2,g_3,f_4,h_5\}$,
  $b\rightarrow \{b,f_1,h_2\}$,
  $b\rightarrow \{b,g_3,f_4,h_5\}$,
  $b\rightarrow \{b,f_4,h_5\}$
}

\noindent The number of states is $|Q|=5$ and the number of transition rules is $|\Delta|=17$. 
\end{example}

\begin{figure}[H]
  \centerline{
	\begin{tikzpicture}[node distance=2.5cm,bend angle=30,transform shape,scale=1]
	  \node[accepting,state,rounded rectangle] (eps) {$\{b,f_1,f_4,g_3,h_2,h_5\}$};
	  \node[state, above of=eps,rounded rectangle,node distance=1.5cm] (h12) {$\{b,f_1,h_2\}$}; 
      \node[state, below of=eps,rounded rectangle] (g13) {$\{b,f_4,g_3,h_5\}$};
      \node[state, right of=g13,rounded rectangle] (g23) {$\{a\}$};
	  \node[state, left of=g13,rounded rectangle] (h15) {$\{b,f_4,h_5\}$};
	  \draw (eps) ++(-2cm,0cm) node {$b$}  edge[->] (eps);  
	  \draw (h12) ++(2cm,0cm) node {$b$}  edge[->] (h12); 
	  \draw (h15) ++(0cm,-1cm) node {$b$}  edge[->] (h15);  
	  \draw (g23) ++(1cm,0cm) node {$a$}  edge[->] (g23);  
	  \draw (g13) ++(0cm,-1cm) node {$b$}  edge[->] (g13);    
      \path[->]
		%
		(h12) edge[->,loop,above] node {$f,h$} ()
		(h12) edge[->,below right] node {$f,h$} (eps)
	    %
		(h15) edge[->, in=135,out=-135,loop,left] node {$f,h$} ()	
		(h15) edge[->,above left] node {$f,h$} (eps)		
		(h15) edge[->,above left] node[pos=0.98] {$f,h$} (g13)		
		%
	  ;
      \draw (eps) ++(1.75cm,-0.75cm)  edge[->,in=0,out=90] node[above right,pos=0] {$g$} (eps) edge (g13) edge (g23);  
      \draw (g13) ++(1.5cm,1cm)  edge[->,in=90,out=145] node[above right,pos=0] {$g$} (g13) edge (g13) edge (g23);
    \end{tikzpicture}
  }
  \caption{The Follow Automaton ${\cal F}_{\E}$.}
  \label{fig r t e3}
\end{figure}
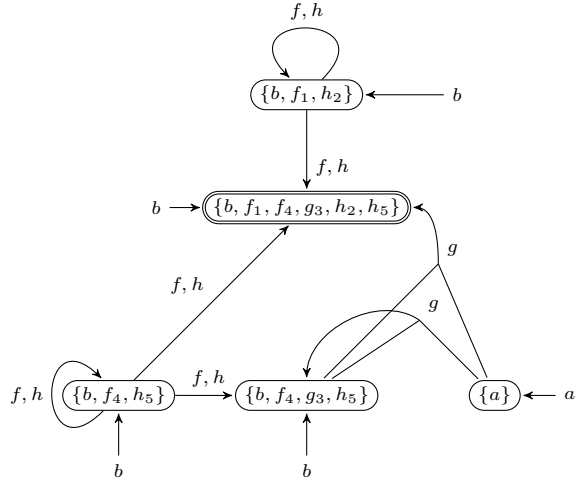

\subsection{The Equation Tree Automaton}

  In \cite{automate2}, Kuske and Meinecke extend the notion of word partial derivatives \cite{antimirov} to tree partial derivatives in order to compute from $\E$ a tree automaton recognizing $\llbracket \E\rrbracket$. 
  Due to the notion of  ranked alphabet, partial derivatives are no longer sets of expressions, but sets of tuples of expressions.
  
  \noindent Let ${\cal N}=(\E_1,\ldots,\E_n)$ be a tuple of regular expressions,
  $\f$ and $\G$ be some regular expressions and $c\in\Sigma_0$. 
  Then ${\cal N}\cdot_c\f$ is the tuple $(\E_1\cdot_c\f,\dots,\E_n\cdot_c\f)$. 
  For  a set ${\cal S}$ of tuples
   of regular expressions, ${\cal S}\cdot_c \f$ is the set ${\cal S}\cdot_c \f=\{{\cal N}\cdot_c\f\mid {\cal N}\in {\cal S}\}$. 
  Finally, $\mathrm{SET}({\cal N})=\{\E_1, \cdots,\E_m\}$ and $\mathrm{SET}({\cal S})=\bigcup_{{\cal N}\in{\cal S}}\mathrm{SET}({\cal N})$.
  
  \noindent Let $f$ be a symbol in $\Sigma_{>0}$. The set $f^{-1}(\E)$ of tuples of regular expressions is defined as follows:
  
  \centerline{ $f^{-1}(0)$ $=\emptyset$,}
  
   \centerline{
   $f^{-1}(F+G)$ $= f^{-1}(\f) \cup  f^{-1}(\G)$,
   }
   
  \centerline{
  $f^{-1}({\f}^{*_c})$ $= f^{-1}({\f})\cdot_{c} {\f}^{*_c}$, }
   
	\centerline{
     $f^{-1}(g(\E_1, \cdots,\E_n))$
     $=\left\{
       \begin{array}{l@{\ }l}
         \{(\E_1, \cdots,\E_n)\} &\text { if } f=g,\\
         \emptyset &\text { otherwise,} 
       \end{array}
     \right.$} 
     
	\centerline{     
     $f^{-1}(\f\cdot_c \G)$
     $=\left\{
      \begin{array}{l@{\ }l}
        f^{-1}(\f)\cdot_c \G&\text { if } c\notin\llbracket\f\rrbracket\\
        f^{-1}(\f)\cdot_c \G\cup f^{-1}(\G) &\text{ otherwise.} 
      \end{array}
     \right.$ 
  }
   
  \noindent The function $f^{-1}$ is extended to any set $S$ of regular expressions by  $f^{-1}(S)=\bigcup_{\E\in S}f^{-1}(\E)$.
  
  \noindent The \emph{partial derivative} of $\E$ w.r.t. a word $w\in\Sigma_{\geq 1}^*$, denoted by $\partial_w(\E)$, is the set of regular expressions inductively defined by: 
  
  \centerline{
   $\partial_w(\E)=
    \left\{
      \begin{array}{l@{\ }l}
        \{\E\} &\text{ if } w=\varepsilon,\\
        \mathrm{SET}(f^{-1}(\partial_{u}(\E)))&\text{ if } w=uf, f\in \Sigma_{\geq 1},u\in\Sigma_{\geq 1}^*, f^{-1}(\partial_{u}(\E))\neq \emptyset,\\
        \{0\}&\text{ if } w=uf, f\in \Sigma_{\geq 1},u\in\Sigma_{\geq 1}^*,f^{-1}(\partial_{u}(\E))=\emptyset.
      \end{array}
    \right.$
}

  
  \noindent The \emph{Equation  Automaton} of $\E$ is the tree automaton ${\cal A_{\E}}=(Q,\Sigma,Q_T,\Delta)$ defined by $Q= \{\partial_{w}(\E)\mid w\in \Sigma^*_{\geq 1}\}$, $Q_T=\{\E\}$, and
  
    
  \centerline{$\begin{array}{r@{\ }c@{\ }l}
    \Delta= & & \{(\f,f,\G_1,\dots,\G_m)\mid \f\in Q,f\in\Sigma_m, {m\geq 1},(\G_1,\dots,\G_m)\in f^{-1}(\f)\}\\
    & \cup & \{(\f,c)\mid~\f\in Q\wedge  c\in(\llbracket \f\rrbracket\cap\Sigma_0)\}\\
      \end{array}$}
    
\begin{example}\label{exp equation automaton}

Let $\E=\underbrace{(f(a)^{*_a}\cdot_a b+ h(b))^{*_b}}_{\f}+\underbrace{g(c,a)^{*_c}}_{\G}\cdot_c \underbrace{(f(a)^{*_a}\cdot_a b+ h(b))^{*_b}}_{\f}$ (Example~\ref{Pos Automat}). 

\centerline{
        $\partial_h(\E)=\{b\cdot_b \f\}$,
        $\partial_f(\E)=\{((a\cdot_af(a)^{*_a})\cdot_a b)\cdot_b \f\}$,
}

\centerline{
  $\partial_{ff}(\E)=\{((a\cdot_af(a)^{*_a})\cdot_a b)\cdot_b \f\}$,
  $\partial_{fh}(\E)=\{b\cdot_b \f\}$
}

\centerline{
  $\partial_g(\E)=\{(a\cdot_c \G)\cdot_c \f, ~(c\cdot_c \G)\cdot_c \f\}$,
  $\partial_{hf}(\E)=\{((a\cdot_af(a)^{*_a})\cdot_a b)\cdot_b \f\}$,
}

\centerline{
  $\partial_{gh}(\E)=\{b\cdot_b \f\}$,
  $\partial_{hh}(\E)=\{((a\cdot_af(a)^{*_a})\cdot_a b)\cdot_b \f\}$
}

\centerline{ 
  $\partial_{gf}(\E)=\{((a\cdot_af(a)^{*_a})\cdot_a b)\cdot_b \f\}$, 
  $\partial_{gg}(\E)=\{(a\cdot_c \G)\cdot_c \f, ~(c\cdot_c \G)\cdot_c \f\}$,
}
 
 The set of states $Q$ is $q_0=\E$, $q_1=((a\cdot_af(a)^{*_a})\cdot_a b)\cdot_b \f$, $q_2=b\cdot_b \f$, $q_3=(c\cdot_c \G)\cdot_c \f$, $q_4=(a\cdot_c \G)\cdot_c \f$. The set of final states is $Q_T=\{q_0\}$. The set of transition rules is
 
 \centerline{
   $b\rightarrow q_0$, $b\rightarrow q_1$, $b\rightarrow q_3$, $b\rightarrow q_2$,
 }
 
 \centerline{$a\rightarrow q_4$,}
 
 \centerline{
   $f(q_1)\rightarrow q_0$,
   $f(q_1)\rightarrow q_1$,
   $f(q_1)\rightarrow q_2$,
   $f(q_1)\rightarrow q_4$,
}

\centerline{
   $h(q_2)\rightarrow q_0$,
   $h(q_2)\rightarrow q_1$,
   $h(q_2)\rightarrow q_2$,
   $h(q_2)\rightarrow q_4$, 
}

\centerline{
   $g(q_3,q_4)\rightarrow q_0$,
   $g(q_3,q_4)\rightarrow q_4$,
}

The number of states is $|Q|=5$ and the number of transition rules is $|\Delta|=15$. 
The Equation Automaton associated with $\E$ is given in Figure~\ref{fig r t e3 eq}.

\end{example}	
	
\begin{figure}[H]
  \centerline{
	\begin{tikzpicture}[node distance=2.5cm,bend angle=30,transform shape,scale=1]
	  \node[accepting,state] (eps) {$q_0$};
      \node[state, below of=eps] (g13) {$q_3$};
      \node[state, right of=g13,node distance=3.5cm] (g23) {$q_4$};
	  \node[state, below left of=g13,node distance=3.5cm] (h15) {$q_2$};
      \node[state, below right of=g13,node distance=3.5cm] (f14) {$q_1$};
	  \draw (eps) ++(-1cm,0cm) node {$b$}  edge[->] (eps);  
	  \draw (h15) ++(0cm,-1cm) node {$b$}  edge[->] (h15);  
	  \draw (g23) ++(1cm,0cm) node {$a$}  edge[->] (g23);  
	  \draw (g13) ++(-1cm,0cm) node {$b$}  edge[->] (g13);    
	  \draw (f14) ++(0cm,-1cm) node {$b$}  edge[->] (f14);
      \path[->]
		(h15) edge[->, in=135,out=-135,loop,left] node {$h$} ()	
		(h15) edge[->,above left] node {$h$} (eps)		
		(h15) edge[->,above left] node {$h$} (g13)		
		(h15) edge[->,bend right,above] node {$h$} (f14)	  
		(f14) edge[->,in=45,out=-45,loop,right] node {$f$} ()	
		(f14) edge[->,bend right,above] node {$f$} (h15)		
		(f14) edge[->,above right] node {$f$} (eps)		
		(f14) edge[->,above right] node {$f$} (g13)
	  ;
      \draw (eps) ++(1.75cm,-0.75cm)  edge[->,in=0,out=90] node[above right,pos=0] {$g$} (eps) edge[] (g13) edge[] (g23);  
      \draw (g13) ++(1.5cm,1cm)  edge[->,in=90,out=145] node[above right,pos=0] {$g$} (g13) edge[] (g13) edge[] (g23);
    \end{tikzpicture}
  }
  \caption{The Equation Automaton ${\cal A}_{\E}$.}
  \label{fig r t e3 eq}
\end{figure}

\subsection{The $k$-C-Continuation Tree Automaton}

In~\cite{automate2}, Kuske and Meinecke show how to efficiently compute the equation tree automaton of a regular expression \emph{via} an extension of Champarnaud and Ziadi's C-Continuation~\cite{ZPC1,ZPC2,khorsi}. 
  In~\cite{cie}, we show how to inductively compute them. 
  We also show how to efficiently compute the $k$-C-Continuation tree automaton associated with a regular expression.
  In this section, we prove that this automaton is isomorphic to the $k$-position tree automaton, similarly to the case of words.

   In this section we consider the following quotient: $0\cdot_c\E=0$.  As we consider only regular expressions without $0$ or reduced to $0$ then if after the computation of $k$-C-Continuation we obtain expression of the form $0\cdot_c\E$ we reduce it to $0$.

  \begin{definition}[\cite{cie}]\label{def1}
   Let $\E\neq 0$ be linear. 
   Let $k$ and $m$ be two integers such that $1\leq k\leq m$. 
   Let $f$ be in $(\Sigma_{\E}\cap\Sigma_{m})$. 
   The \emph{$k$-C-continuation} $C_{f^k}(\E)$ of $f$ in $\E$ is the regular expression defined by:
    
\centerline{
  \begin{tabular}{r@{\ }l}
    $C_{f^k}(g(\E_1, \cdots,\E_m))$ & 
      $=\left\{
        \begin{array}{l@{\ \  \  \ \  \  \  }l}
          \E_k & \text { if } f=g\\
          C_{f^k}(\E_j)& \text{ if } f\in \Sigma_{\E_j}\\
        \end{array}
      \right.$\\
    $C_{f^k}(\E_1+\E_2)$ & 
    $=\left\{
      \begin{array}{l@{\ \  \  \ \  \  \  }l}
        C_{f^k}(\E_1)& \text{ if } f\in \Sigma_{\E_1}\\
        C_{f^k}(\E_2)& \text{ if } f\in \Sigma_{\E_2}\\
      \end{array}
    \right.$\\
    $C_{f^k}(\E_1\cdot_c \E_2)$ & 
    $=\left\{ 
      \begin{array}{l@{\ }l}
        C_{f^k}(\E_1)\cdot_{c} \E_2 &\text{ if } f\in \Sigma_{\E_1}\\
        C_{f^k}(\E_2)             &\text{ if } f\in \Sigma_{\E_2} \\
        & \text{ and } c\in \Last(E_1)\\
        0 					   &\text{ otherwise }\\
      \end{array}
    \right.$\\
    $C_{f^k}({\E_1}^{*_c})$ & $=C_{f^k}({\E_1})\cdot_{c} {\E_1}^{*_c}$\\
  \end{tabular}
}

     By convention, we set $C_{\varepsilon^1}(\E)=\E$. 
  \end{definition}
\begin{lemma}
Let $\E$ be a regular expression without occurrences of $0$ or reduced to $0$. Then, $C_{f^k}({\E})$ is a regular expression without occurrences of $0$ or reduced to $0$.  
\end{lemma}
\begin{proof}
\begin{sloppy}
Let us show by induction over $\E$ that if $\E$ is a regular expression without occurrences of $0$ or reduced to $0$ then, $C_{f^k}({\E})$ is a regular expression without occurrences of $0$ or reduced to $0$.     
 Suppose that $\E=g(\E_1,\dots,\E_m)$. Hence by definition $C_{f^k}(\E)=\E_k$ and $\E$ is a regular expression without occurrences of $0$ or reduced to $0$. The property is true for the base case.

  Assuming that the property holds for the subexpressions of $\E$.
  \begin{enumerate}
\item Consider that $\E=g(\E_1,\dots,\E_m)$ with $f\neq g$. Then by definition $C_{f^k}(\E)=\E_j$ with $f\in\Sigma_{\E_j}$. By induction hypothesis, $C_{f^k}(\E)=\E_j$ is a regular expression without occurrences of $0$ or reduced to $0$.
\item Let us consider that $\E=\E_1+\E_2$. Suppose that $f\in\Sigma_{\E_p}$ with $p\in\{1,2\}$. Hence From Definition~\ref{def1},$C_{f^k}(\E_1+\E_2)=C_{f^k}(\E_p)$. By induction hypothesis, $C_{f^k}(\E_1+\E_2)$ is a regular expression without occurrences of $0$ or reduced to $0$.   
  \item Consider that $\E=\E_1 \cdot_c \E_2$. Three cases may occur.
   \begin{enumerate}
  \item Suppose that $f\in \Sigma_{\E_1}$. Then by Definition~\ref{def1}, $C_{f^k}(\E_1\cdot_c \E_2)=C_{f^k}(\E_1)\cdot_{c} \E_2$. By induction hypothesis, $C_{f^k}(\E_1)$ and $\E_2$ are regular expressions without occurrences of $0$ or reduced to $0$.  Therefore, $C_{f^k}(\E_1\cdot_c \E_2)=C_{f^k}(\E_1)\cdot_{c} \E_2$ is a regular expression without occurrences of $0$ or reduced to $0$.
\item Consider that $f\in \Sigma_{\E_2}$ and $c\in \Last(E_1)$. In this case $C_{f^k}(\E_1\cdot_c \E_2)=C_{f^k}(\E_2)$. By induction hypothesis, $C_{f^k}(\E_1\cdot_c \E_2)=C_{f^k}(\E_2)$ is a regular expression without occurrences of $0$ or reduced to $0$.   
      \item Consider that $f\in \Sigma_{\E_2}$ and $c\notin \Last(E_1)$. In this case $C_{f^k}(\E_1\cdot_c \E_2)=0$. The property is verified.
      \end{enumerate}
   \item Consider that $\E={\E_1}^{*_c}$. By Definition~\ref{def1}, $C_{f^k_j}(\E_1^{*_c})=C_{f^k_j}(\E_1)\cdot_c \E_2$. By induction hypothesis, $C_{f^k_j}(\E_1)$ and $\E_2$ are regular expressions without occurrences of $0$ or reduced to $0$. Therefore, $C_{f^k_j}(\E_1^{*_c})=C_{f^k_j}(\E_1)\cdot_c \E_2$ is a regular expression without occurrences of $0$ or reduced to $0$. 
  \end{enumerate}
  \end{sloppy} 
 \end{proof}
 
  Let us now show how to compute the $k$-C-Continuation tree automaton.

\begin{definition}[\cite{cie}]\label{def aut c cont lin}
  Let $\E\neq 0$ be linear. 
 The automaton ${\cal  C_{\E}}=(Q_{{\cal C}},\Sigma_E,\{C_{{\varepsilon}^1}(\E)\},$ $\Delta_{{\cal C}})$ is defined by  
 
 \begin{tabular}{l@{\ }l}
 $Q_{{\cal C}}=$ & $\{(f^k,C_{f^k}(\E))\mid f\in \Sigma_m,1\leq k\leq m\}\cup\{({\varepsilon}^1,C_{{\varepsilon}^1}(\E))\}$,\\
  $\Delta_{{\cal C}}=$ & $\{((x,C_{x}(\E)),g,((g^1,C_{g^1}(\E)),\dots,(g^m,C_{g^m}(\E))))\mid  g\in{\Sigma_E}_m,$\\
      & $m\geq 1, (C_{g^1}(\E),\dots,C_{g^m}(\E))\in {g}^{-1}(C_{x}(\E))\}$\\
      & $\cup \{((x,C_{x}(\E)),c)\mid, c\in\llbracket C_{x}(\E)\rrbracket\cap\Sigma_0\}$\\
    \end{tabular}
\end{definition}

The \emph{C-Continuation tree automaton} ${\cal C_{\E}}$ associated with $\E$ is obtained by relabelling the transitions of ${\cal  C_{\b\E}}$ using the mapping $h$.

\begin{theorem}[\cite{cie}]
 The automaton ${\cal  C_{\E}}$ accepts $\llbracket \E\rrbracket$.
\end{theorem}

\begin{example}
Let $\E=(f(a)^{*_a}\cdot_a b+ h(b))^{*_b}+g(c,a)^{*_c}\cdot_c (f(a)^{*_a}\cdot_a b+ h(b))^{*_b}$ defined in Example~\ref{Pos Automat} and 
$\b\E=\underbrace{(f_1(a)^{*_a}\cdot_a b+ h_2(b))^{*_b}}_{\f_1}+\underbrace{g_3(c,a)^{*_c}}_{\G_2}\cdot_c \underbrace{(f_4(a)^{*_a}\cdot_a b+ h_5(b))^{*_b}}_{\f_3}$. 

The computation of the $k$-C-Continuations of $\E$ using the Definition~\ref{def1} is given in Table~\ref{tab c-cont}.


\begin{table}[H]
\centerline{
    \begin{tabular}{l@{\ }ll@{\ }l} 
$C_{f^1_1}(\b\E)=$ & $((a\cdot_a f_1(a)^{*_a})\cdot_a b)\cdot_b \f_1$ & $h(C_{f^1_1}(\b\E))=$  & $((a\cdot_a f(a)^{*_a})\cdot_a b)\cdot_b \f$,\\
$C_{h^1_2}(\b\E)=$ & $b\cdot_b \f_1$ & $h(C_{h^1_2}(\b\E))=$  & $b\cdot_b \f$, \\
$C_{g^1_3}(\b\E)=$ & $(c\cdot_c g_3(c,a)^{*_c})\cdot_c \f_3$ &  $h(C_{g^1_3}(\b\E))=$ & $(c\cdot_c g(c,a)^{*_c})\cdot_c \f$,\\
$C_{g^2_3}(\b\E)=$ & $(a\cdot_c g_3(c,a))^{*_c})\cdot_c \f_3$ & $h(C_{g^2_3}(\b\E))=$ &  $(a\cdot_c g(c,a))^{*_c})\cdot_c \f$,\\
$C_{f^1_4}(\b\E)=$ & $((a\cdot_a f_4(a)^{*_a})\cdot_a b)\cdot_b \f_3$ & $h(C_{f^1_4}(\b\E))=$ & $((a\cdot_a f(a)^{*_a})\cdot_a b)\cdot_b \f$,\\ 
$C_{h^1_5}(\b\E)=$ & $b\cdot_b \f_3$ &	$h(C_{h^1_5}(\b\E))=$  & $b\cdot_b \f$.\\
   \end{tabular}
  }
  \caption{The $k$-C-Continuations of $\b\E$.}
  \label{tab c-cont}
\end{table} 

The set of states of the automaton ${\cal C_{\E}}$ is $Q=\{(\varepsilon^1,C_{\varepsilon^1}(\b\E)),(f^1_1,C_{f^1_1}(\b\E)),(h^1_2,C_{h^1_2}(\b\E)),\\(g^1_3,C_{g^1_3}(\b\E)),(g^2_3,C_{g^2_3}(\b\E)),
(f^1_4,C_{f^1_4}(\b\E)),(h^1_5,C_{h^1_5}(\b\E))\}$.

 The set of transition rules $\Delta$ is
 
 \centerline{
   $a \rightarrow  (g^2_3,C_{g^2_3}(\b\E))$,
}

 \centerline{
   $b\rightarrow (f^1_1,C_{f^1_1}(\b\E))$,\ \ 
   $b\rightarrow (g^1_3,C_{g^1_3}(\b\E))$,\ \ 
   $b\rightarrow  (\varepsilon^1, C_{\varepsilon^1}(\E))$,
}

 \centerline{
   $b\rightarrow (h^1_2,C_{h^1_2}(\b\E))$,\ \ 
   $b\rightarrow (h^1_5,C_{h^1_5}(\b\E))$,\ \ 
   $b\rightarrow (f^1_4,C_{f^1_4}(\b\E))$,
}

 \centerline{
   $f( (f^1_4,C_{f^1_4}(\b\E)) )\rightarrow (g^1_3,C_{g^1_3}(\b\E))$,\ \ 
   $f((f^1_4,C_{f^1_4}(\b\E)) )\rightarrow(f^1_4,C_{f^1_4}(\b\E))$,
}

 \centerline{
   $f((f^1_4,C_{f^1_4}(\b\E)) )\rightarrow(h^1_5,C_{h^1_5}(\b\E))$,\ \ 
   $f( (f^1_1,C_{f^1_1}(\b\E)) )\rightarrow(h^1_2,C_{h^1_2}(\b\E))$,
}

 \centerline{
   $f((f^1_1,C_{f^1_1}(\b\E)) )\rightarrow(\varepsilon^1,C_{\varepsilon^1}(\E))$,\ \ 
   $f( (f^1_1,C_{f^1_1}(\b\E)) )\rightarrow (f^1_1,C_{f^1_1}(\b\E))$,
}

 \centerline{
   $f((f^1_4,C_{f^1_4}(\b\E))) \rightarrow (\varepsilon^1,C_{\varepsilon^1}(\E))$,
}

 \centerline{
   $g((g^1_3,C_{g^1_3}(\b\E)),(g^2_3,C_{g^2_3}(\b\E))) \rightarrow  (\varepsilon^1, C_{\varepsilon^1}(\E))$,
}

 \centerline{
   $g((g^1_3,C_{g^1_3}(\b\E)), (g^2_3,C_{g^2_3}(\b\E)) ) \rightarrow (g^1_3,C_{g^1_3}(\b\E))$,
}

 \centerline{
   $h((h^1_5,C_{h^1_5}(\b\E)) )\rightarrow (\varepsilon^1, C_{\varepsilon^1}(\E)$,\ \ 
   $h( (h^1_2,C_{h^1_2}(\b\E)) )\rightarrow(h^1_2,C_{h^1_2}(\b\E))$,
}

 \centerline{
   $h((h^1_2,C_{h^1_2}(\b\E)))\rightarrow(\varepsilon^1,C_{\varepsilon^1}(\E))$,\ \ 
   $h( (h^1_2,C_{h^1_2}(\b\E)) )\rightarrow (f^1_1,C_{f^1_1}(\b\E))$,
}

 \centerline{
   $h((h^1_5,C_{h^1_5}(\b\E)) )\rightarrow (h^1_5,C_{h^1_5}(\b\E))$,\ \ 
   $h((h^1_5,C_{h^1_5}(\b\E)) )\rightarrow (f^1_4,C_{f^1_4}(\b\E))$,
}

 \centerline{
   $h( (h^1_5,C_{h^1_5}(\b\E)) ) \rightarrow (g^1_3,C_{g^1_3}(\b\E))$
}

\noindent The number of states is $|Q|=7$ and the number of transition rules is $|\Delta|=23$.  
The $k$-C-Continuation Automaton associated with $\E$ is given in Figure~\ref{fig Cfk e}.  
\end{example}

\begin{figure}[H]
  \centerline{
	\begin{tikzpicture}[node distance=2.5cm,bend angle=30,transform shape,scale=1]
	  \node[accepting,state,rounded rectangle] (eps) {$(\varepsilon^1,C_{\varepsilon^1}(\b E))$};
	  \node[state, above left of=eps,rounded rectangle] (f11) {$(f^1_1,C_{f^1_1}(\b E))$};	
	  \node[state, above right of=eps,rounded rectangle] (h12) {$(h^1_2,C_{h^1_2}(\b E))$}; 
      \node[state, below of=eps,rounded rectangle] (g13) {$(g^1_3,C_{g^1_3}(\b E))$};
      \node[state, right of=g13,node distance=3.5cm,rounded rectangle] (g23) {$(g^2_3,C_{g^2_3}(\b E))$};
	  \node[state, below left of=g13,node distance=3.5cm,rounded rectangle] (h15) {$(h^1_5,C_{h^1_5}(\b E))$};
      \node[state, below right of=g13,node distance=3.5cm,rounded rectangle] (f14) {$(f^1_4,C_{f^1_4}(\b E))$};
	  \draw (eps) ++(-2cm,0cm) node {$b$}  edge[->] (eps);  
	  \draw (f11) ++(-2cm,0cm) node {$b$}  edge[->] (f11);  
	  \draw (h12) ++(2cm,0cm) node {$b$}  edge[->] (h12); 
	  \draw (h15) ++(0cm,-1cm) node {$b$}  edge[->] (h15);  
	  \draw (g23) ++(2cm,0cm) node {$a$}  edge[->] (g23);  
	  \draw (g13) ++(-2cm,0cm) node {$b$}  edge[->] (g13);    
	  \draw (f14) ++(0cm,-1cm) node {$b$}  edge[->] (f14);
      \path[->]
        (f11) edge[->,below left] node {$f_1$} (eps)
		(f11) edge[->,loop,above] node {$f_1$} ()
		(h12) edge[->,bend right,above] node {$h_2$} (f11)
		%
		(h12) edge[->,loop,above] node {$h_2$} ()
		(h12) edge[->,below right] node {$h_2$} (eps)
		(f11) edge[->,bend right,above] node {$f_1$} (h12)
	    %
		(h15) edge[->, in=135,out=-135,loop,left] node {$h_5$} ()	
		(h15) edge[->,above left] node {$h_5$} (eps)		
		(h15) edge[->,above left] node {$h_5$} (g13)		
		(h15) edge[->,bend right,above] node {$h_5$} (f14)	  
		(f14) edge[->,in=45,out=-45,loop,right] node {$f_4$} ()	
		(f14) edge[->,bend right,above] node {$f_4$} (h15)		
		(f14) edge[->,above right] node {$f_4$} (eps)		
		(f14) edge[->,above right] node {$f_4$} (g13)
	  ;
      \draw (eps) ++(1.75cm,-0.75cm)  edge[->,in=0,out=90] node[above right,pos=0] {$g_3$} (eps) edge[] (g13) edge[] (g23);  
      \draw (g13) ++(1.5cm,1cm)  edge[->,in=90,out=145] node[above right,pos=0] {$g_3$} (g13) edge[] (g13) edge[] (g23);
    \end{tikzpicture}
  }
  \caption{The $k$-C-Continuation Automaton ${\cal C}_{\E}$.}
  \label{fig Cfk e}
\end{figure}

	Let $\sim_e$ be the equivalence relation over the set of states of ${\cal  C_{\E}}$ defined for any two states $(f^k_j,C_{f^k_j}(\b\E))$ and $(g^p_i,C_{g^p_i}(\b\E))$ by $(f^k_j,C_{f^k_j}(\b\E)) \sim_e (g^p_i,C_{g^p_i}(\b\E)) \Leftrightarrow h(C_{f^k_j}(\b\E))=h(C_{g^p_i}(\b\E))$.
	
	\begin{proposition}[\cite{cie}]
	  The automaton ${\cal C_{\E}}\diagup_{\sim_e}$ is isomorphic to ${\cal A_{\E}}$.
	\end{proposition}

\begin{example}
Using the equivalence-relation $\sim_e$ over the set of states of $k$-C-Continuation Automaton ${\cal C}_{\E}$ (Figure~\ref{fig Cfk e}) we see that $h(C_{f^1_1}(\b\E))=h(C_{f^1_4}(\b\E))$ and $h(C_{h^1_2}(\b\E))=h(C_{h^1_5}(\b\E))$. The automaton ${\cal C_{\E}}\diagup_{\sim_e}$	 is given in Figure~\ref{fig r t e31}. The number of states is $|Q|=5$ and the number of transition rules is $|\Delta|=15$.

\begin{figure}[H]
  \centerline{
	\begin{tikzpicture}[node distance=2.5cm,bend angle=30,transform shape,scale=1]
	  \node[accepting,state,rounded rectangle] (eps) {$\{h(C_{\varepsilon^1}(E))\}$};
      \node[state, below of=eps,rounded rectangle] (g13) {$\{h(C_{g_3^1}(E))\}$};
      \node[state, right of=g13,node distance=3.5cm,rounded rectangle] (g23) {$\{h(C_{g_3^2}(E))\}$};
	  \node[state, below left of=g13,node distance=3.5cm,rounded rectangle] (h15) {$\{h(C_{h_2^1}(E))\}$};
      \node[state, below right of=g13,node distance=3.5cm,rounded rectangle] (f14) {$\{h(C_{f_1^1}(E))\}$};
	  \draw (eps) ++(-2cm,0cm) node {$b$}  edge[->] (eps);  
	  \draw (h15) ++(0cm,-1cm) node {$b$}  edge[->] (h15);  
	  \draw (g23) ++(2cm,0cm) node {$a$}  edge[->] (g23);  
	  \draw (g13) ++(-2cm,0cm) node {$b$}  edge[->] (g13);    
	  \draw (f14) ++(0cm,-1cm) node {$b$}  edge[->] (f14);
      \path[->]
		(h15) edge[->, in=135,out=-135,loop,left] node {$h$} ()	
		(h15) edge[->,above left] node {$h$} (eps)		
		(h15) edge[->,above left] node {$h$} (g13)		
		(h15) edge[->,bend right,above] node {$h$} (f14)	  
		(f14) edge[->,in=45,out=-45,loop,right] node {$f$} ()	
		(f14) edge[->,bend right,above] node {$f$} (h15)		
		(f14) edge[->,above right] node {$f$} (eps)		
		(f14) edge[->,above right] node {$f$} (g13)
	  ;
      \draw (eps) ++(1.75cm,-0.75cm)  edge[->,in=0,out=90] node[above right,pos=0] {$g$} (eps) edge[] (g13) edge[] (g23);  
      \draw (g13) ++(1.5cm,1cm)  edge[->,in=90,out=145] node[above right,pos=0] {$g$} (g13) edge[] (g13) edge[] (g23);
    \end{tikzpicture}
  }
  \caption{The Automaton ${\cal C_{\E}}\diagup_{\sim_e}$.}
  \label{fig r t e31}
\end{figure}	
\end{example}	

In order to show that the $k$-C-continuation tree automaton of $\E$ is isomorphic to the $k$-position automaton of $\E$, we first show the link between the position functions and the C-continuations.
  
  \begin{proposition}[\cite{cie}]\label{prop CfkFollow}
    Let $\E$ be linear, $1\leq k\leq m$ be two integers and $f$ be a position 
    in $\Sigma_E\cap \Sigma_m$. 
    Then $\Follow(E,f,k)=\First(C_{f^k}(\b{\E}))$. 
  \end{proposition}
  
  \begin{lemma}\label{lemma FirstContinuation}
    Let $\E$ be linear and $g$ be a symbol in $\Sigma_{\geq 1}$. 
    Then $g^{-1}(\E)\neq \emptyset\Leftrightarrow g\in \First(\E)$.
  \end{lemma}
  \begin{proof}
    By induction on the construction of $\E$.
    \begin{enumerate}
      \item If $E=a$, then $g^{-1}(\E)=\emptyset$ and $g\notin\First(\E)=\{a\}$.
      \item Suppose that $E=f(E_1,\ldots,E_n)$. 
        Hence $\mathrm{First}(E)=\{f\}$.
        Then:
        
        \centerline{
          \begin{tabular}{l@{\ }l}
            $g^{-1}(\E)\neq \emptyset$ & $\Leftrightarrow$ $f=g$\\
            & $\Leftrightarrow$ $g\in \mathrm{First}(E)$\\
          \end{tabular}
        }
      \item Consider that $E=E_1+E_2$.
        Then:
        
        \centerline{
          \begin{tabular}{l@{\ }l}
            $g^{-1}(\E)\neq \emptyset$ & $\Leftrightarrow$ $g^{-1}(\E_1)\neq \emptyset \vee g^{-1}(\E_2)\neq \emptyset$\\
            & $\Leftrightarrow$ $g\in \mathrm{First}(E_1)\vee g\in \mathrm{First}(E_2)$\ \ \ \ \textbf{(Induction Hypothesis)}\\
            & $\Leftrightarrow$ $g\in \mathrm{First}(E)$\\
          \end{tabular}
        }
      \item Suppose that $E=E_1\cdot_c E_2$. 
        Two cases have to be considered.
        \begin{enumerate}
          \item Suppose that $c\in \llbracket E_1\rrbracket$. 
            Then:
            
            \centerline{
              \begin{tabular}{l@{\ }l}
                $g^{-1}(\E)\neq \emptyset$ & $\Leftrightarrow$ $g^{-1}(\E_1)\cdot_c E_2 \neq \emptyset \vee g^{-1}(\E_2)\neq \emptyset$\\
                & $\Leftrightarrow$ $g^{-1}(\E_1) \neq \emptyset \vee g^{-1}(\E_2)\neq \emptyset$\\
                & $\Leftrightarrow$ $g\in \mathrm{First}(E_1)\vee g\in \mathrm{First}(E_2)$\ \ \ \ \textbf{(Induction Hypothesis)}\\
                & $\Leftrightarrow$ $g\in \mathrm{First}(E_1)\setminus\{c\}\vee g\in \mathrm{First}(E_2)$\ \ \ \ \textbf{($g\neq c$)}\\
                & $\Leftrightarrow$ $g\in \mathrm{First}(E)$\\
              \end{tabular}
            }
          \item  Consider that $c\notin \llbracket E_1\rrbracket$.  
            Then:
            
            \centerline{
              \begin{tabular}{l@{\ }l}
                $g^{-1}(\E)\neq \emptyset$ & $\Leftrightarrow$ $g^{-1}(\E_1)\cdot_c E_2 \neq \emptyset$\\
                & $\Leftrightarrow$ $g^{-1}(\E_1) \neq \emptyset$\\
                & $\Leftrightarrow$ $g\in \mathrm{First}(E_1)$\ \ \ \ \textbf{(Induction Hypothesis)}\\
                & $\Leftrightarrow$ $g\in \mathrm{First}(E)$\\
              \end{tabular}
            }
        \end{enumerate}
      \item Suppose that $E=E_1^{*_c}$. 
        Then:
        
        \centerline{
          \begin{tabular}{l@{\ }l}
            $g^{-1}(\E)\neq \emptyset$ & $\Leftrightarrow$ $g^{-1}(\E_1)\cdot_c E_1^{*_c}\neq \emptyset$\\
            & $\Leftrightarrow$ $g^{-1}(\E_1) \neq \emptyset$\\
            & $\Leftrightarrow$ $g\in \mathrm{First}(E_1)$\ \ \ \ \textbf{(Induction Hypothesis)}\\
            & $\Leftrightarrow$ $g\in \mathrm{First}(E)$\\
          \end{tabular}
        }
    \end{enumerate}
  \end{proof}
  
  \begin{corollary}\label{corol FirstContinuation}
    Let $\E$ be linear, $1\leq k\leq m$ be two integers and $f$ and $g$ be two symbols in $\Sigma$. 
    Then, $g^{-1}(C_{f^k}(\E))\neq \emptyset\Leftrightarrow g\in \First(C_{f^k}(\E))$.
  \end{corollary}
  
  \begin{lemma}\label{lemma FollowContinuation}
    Let $\E$ be linear, $1\leq k\leq m$ be two integers and $f$ and $g$ be two symbols in $\Sigma$. 
    Then, $g^{-1}(C_{f^k}(\E))\neq \emptyset\Leftrightarrow g\in \Follow(\E,f,k)$. 
  \end{lemma}
  \begin{proof}
    According to Corollary~\ref{corol FirstContinuation}, $g^{-1}(C_{f^k}(\E))\neq \emptyset\Leftrightarrow g\in \First(C_{f^k}(\E))$. 
    Following Proposition~\ref{prop CfkFollow}, $\First(C_{f^k}(E))$ $=\Follow(F,f,k)$.
    Consequently, it holds that $g^{-1}(C_{f^k}(\E))\neq \emptyset\Leftrightarrow g\in \Follow(F,f,k)$.
  \end{proof} 
  
  \begin{proposition}\label{prop c cont aut iso position}
    Let $E$ be linear.
    The automaton ${\cal C_{\E}}$ is isomorphic to ${\cal P_{\E}}$.
  \end{proposition}
  \begin{proof}
    Let ${\cal C_{\E}}=(Q,\Sigma,Q_T,\Delta)$ and  ${\cal P_{\E}}=(Q',\Sigma,Q'_T,\Delta')$.
    Let $\phi$ be the bijection between $Q$ and $Q'$ defined for any state $q_{f_k}=(f^k,C_{f^k}(E))$ in $Q$ by $\phi(q_{f_k})=f_k$.
    Then for any symbol $f$ in $\Sigma_m$:
    
    \centerline{
      \begin{tabular}{l@{\ }l}
        & $((g^l,C_{g^l}(E)),f,(f^1,C_{f^1}(E)),\ldots,(f^m,C_{f^m}(E)))\in\Delta$\\
        $\Leftrightarrow$ & $(C_{f^1}(E),\ldots,C_{f^m}(E))\in f^{-1}(C_{g^l}(E))$\ \ \ \ \textbf{(Definition of ${\cal C_{\E}}$)}\\
        $\Leftrightarrow$ & $f^{-1}(C_{g^l}(E))\neq\emptyset$\\
        $\Leftrightarrow$ & $f\in \mathrm{Follow}(E,g,l)$\ \ \ \ \textbf{(Lemma~\ref{lemma FollowContinuation})}\\
        $\Leftrightarrow$ & $(g^l,f,f^1,\ldots,f^m)\in\Delta'$\ \ \ \ \textbf{(Construction of ${\cal P_{\E}}$)}\\
        $\Leftrightarrow$ & $(\phi((g^l,C_{g^l}(E)),f,\phi((f^1,C_{f^1}(E))),\ldots,\phi((f^m,C_{f^m}(E))))\in\Delta'$\\
      \end{tabular} 
    }
  \end{proof}
  
  This proposition can be extended to 
  expressions that are not necessarily linear  
  since ${\cal C_{\E}}$ and ${\cal P_{\E}}$ are 
  relabelings
  of ${\cal C_{\b\E}}$ and ${\cal P_{\b\E}}$.
  
  \begin{corollary}
    The automaton ${\cal C_{\E}}$ is isomorphic to ${\cal P_{\E}}$.
  \end{corollary}
  
  We define the 
  similarity
  relation denoted by $\equiv$ over the set of states of the automaton ${\cal C_{\E}}$ as follows:
  
  \centerline{$(f^k,C_{f^k}(\b\E))\equiv(g^p,C_{g^p}(\b\E))\Leftrightarrow \Follow(\b\E,f,k)=\Follow(\b\E,g,p)$.}
  
  \begin{corollary}
    The finite tree automaton ${\cal C_{\E}}\diagup_{\equiv}$ is isomorphic to the follow automaton ${\cal F}_{\E}$.
  \end{corollary}

\section{Comparison between the Equation and the Follow Automata}\label{sec compar}

We discuss in this section two examples to compare the equation and the follow automata.

 Let $\Sigma=\Sigma_0\cup\Sigma_1$ be the ranked alphabet defined by $\Sigma_0=\{a\}$ and $\Sigma_1=\{f_1,\ldots,f_n\}$. 
    Let us consider the linear regular expression $\E=((f_1(a)^{*_a}\cdot_a f_2(a)^{*_a})\cdot_a \ldots)\cdot_a f_n(a)^{*_a}))^{*_a}$ defined over $\Sigma$. Then the size of $\E$ is $|\E|=4n-1$ and its alphabet width is $||\E|| = n+1$. We have 
$\First(\E)=\{a,f_1,f_2,\ldots,f_n\}$ and $\Follow(\E,f_1,1)=\Follow(\E,f_2,1)=\ldots=\Follow(\E,f_n,1)=\{a,f_1,f_2,\ldots,f_n\}$. 

\noindent The partial derivatives associated with $\E$ are:

\centerline{$\partial_{f_1}(\E)=\{(((a\cdot_a f_1(a)^{*_a}\cdot_a f_2(a)^{*_a})\cdot_a \ldots)\cdot_a f_n(a)^{*_a})\cdot_a \E\}$}

\centerline{$\partial_{f_2}(\E)=\{((a\cdot_a f_2(a)^{*_a}\cdot_a \ldots)\cdot_a f_n(a)^{*_a})\cdot_a \E\}$, $\ldots$}

\centerline{$\partial_{f_n}(\E)=\{(a \cdot_a f_n(a)^{*_a})\cdot_a \E\}$.}

\noindent The $K$-position automaton associated with $\E$ has $n+1$ states. 

\noindent The follow automaton associated with $\E$ has $1$ state.

\noindent The equation automaton associated with $\E$ has: $n+1$ states.

Let $\f=\underbrace{(f(a)^{*_a}+f(a)^{*_a}+\cdots +f(a)^{*_a})}_{f(a)^{*_a} ~n\mbox{-times} }$ be a regular expression defined over the ranked alphabet $\Sigma=\Sigma_0\cup\Sigma_1$ such that $\Sigma_0=\{a\}$ and $\Sigma_1=\{f\}$. We have $|\E|=4n-1$ and $||\E||= n+1$. The linearized form associated with $\f$ is $\b\f=(f_1(a)^{*_a}+f_2(a)^{*_a}+\cdots +f_n(a)^{*_a})$. The set $\First(\f)=\{a,f_1,f_2,\ldots,f_n\}$, $\Follow(\b\f,f_1,1)=\{a,f_1\}$, $\Follow(\b\f,f_2,1)=\{a,f_2\},\ldots,$ and  $\Follow(\b\f,f_n,1)=\{a,f_n\}$.

\noindent The partial derivatives associated with $\f$ are $\partial_{f}(\f)=\{a\cdot_a f(a)^{*_a}\}$, $\partial_{ff}(\f)=\{a\cdot_a f(a)^{*_a}\}$.

\noindent The $K$-position automaton associated with $\f$ has $n+1$ states. 

\noindent The follow automaton associated with $\f$ has: $n+1$ states.

\noindent The equation automaton associated with $\f$ has: $2$ states.

From these examples we state that the two automata are incomparable: 
 
\begin{proposition}
The Follow tree automaton
 and the Equation Tree Automaton are incomparable though they 
 are derived from two isomorphic automata,
  \emph{i.e.} Neither is a quotient of the other.
\end{proposition}

\subsection{A smaller automaton}
In \cite{garcia} P. Garc\'{\i}a ~\emph{et al.} proposed an algorithm to obtain an automaton from a word regular expression. 
Their method is based 
on
 the computation of both the partial derivatives automaton and the follow automaton. They join two relations, the first 
relation
 is over the states of 
 the word
  follow automaton and the second 
 relation
 is over the word c-continuations automaton, in one relation denoted by $\equiv_V$. 
What we propose is to extend the relation $\equiv_V$ to the case of trees as 
follows:

\centerline{
$C_{f^k_j}(\b{\E}) \equiv_V C_{g^p_i}(\b{\E})$ $\Leftrightarrow$ 
$\left\{
  \begin{array}{l@{\ }l}
    & (\exists C_{h^l_m}(\b{\E}) \sim_{\cal F} C_{f^k_j}(\b{\E}) \mid~ C_{h^l_m}(\b{\E})\sim_e C_{g^p_i}(\b{\E}))\\
    \vee & (\exists C_{h^l_m}(\b{\E})\sim_{\cal F} C_{g^p_i}(\b{\E}))\mid~ C_{h^l_m}(\b{\E})\sim_e  C_{f^k_j}(\b{\E}))
  \end{array}
\right.$
}

The idea is to define the follow relation $\sim_{\cal F}$ over the states of the c-continuation automaton ${\cal C_{\E}}$ as follows: 
	
\centerline{$C_{f_j^k}(\b{\E}) \sim_{\cal F} C_{g_i^p}(\b{\E})\Leftrightarrow\Follow(C_{f_j^k}(\b{\E}),f_j,k)=\Follow(C_{g_i^p}(\b{\E}),g_i,p)$}
    
such that we keep all the equivalent $k$-C-Continuations in the merged states.
The obtained automaton is denoted by ${\cal C_{\E}}\diagup_{\sim_{\cal F}}$.  
Then apply the 
equivalence 
 relation $\sim_e$ (apply the mapping $h$) over the states of the automaton ${\cal C_{\E}}\diagup_{\sim_{\cal F}}$ and merge the states which have at least one expression in common. 

\begin{example} 
Let $\E=(f(a)^{*_a}\cdot_a b+ h(b))^{*_b}+g(c,a)^{*_c}\cdot_c (f(a)^{*_a}\cdot_a b+ h(b))^{*_b}$ defined in Example~\ref{Pos Automat} and  $\b\E=\underbrace{(f_1(a)^{*_a}\cdot_a b+ h_2(b))^{*_b}}_{\f_1} +\underbrace{g_3(c,a)^{*_c}}_{\G_2}\cdot_c \underbrace{(f_4(a)^{*_a}\cdot_a b+ h_5(b))^{*_b}}_{\f_3}$. 

\centerline{
    \begin{tabular}{lll}
$C_{f^1_1}(\b\E)=$ $((a\cdot_a f_1(a)^{*_a})\cdot_a b)\cdot_b (f_1(a)^{*_a}\cdot_a b+ h_2(b))^{*_b}$, &\\
$C_{h^1_2}(\b\E)=$  $b\cdot_b  (f_1(a)^{*_a}\cdot_a b+ h_2(b))^{*_b}$,\\
$C_{g^1_3}(\b\E)=$ $(c\cdot_c g_3(c,a)^{*_c})\cdot_c  (f_4(a)^{*_a}\cdot_a b+ h_5(b))^{*_b}$,&\\ 
$C_{g^2_3}(\b\E)=$ $(a\cdot_c g_3(c,a))^{*_c})\cdot_c (f_4(a)^{*_a}\cdot_a b+ h_5(b))^{*_b}$, \\
$C_{f^1_4}(\b\E)=$  $((a\cdot_a f_4(a)^{*_a})\cdot_a b)\cdot_b (f_4(a)^{*_a}\cdot_a b+ h_5(b))^{*_b}$,&\\
$C_{h^1_5}(\b\E)=$ $b\cdot_b (f_4(a)^{*_a}\cdot_a b+ h_5(b))^{*_b}$. \\
   \end{tabular}
  }     

Applying $\sim_{\cal F}$ over the states of ${\cal C_{\E}}$ 
we obtain:  
$C_{f^1_1}(\b\E)\sim_{\cal F} C_{h^1_2}(\b\E)$ then the two states are merged, $C_{f^1_4}(\b\E)\sim_{\cal F} C_{h^1_5}(\b\E)$ so they are merged. The states $C_{g^1_3}(\b\E)$ and $C_{g^2_3}(\b\E)$ are not merged with anyone.  

\noindent The number of states is $|Q|=5$ and the number of transition rules is $|\Delta|=15$.
The quotient automaton of this automaton by the 
equivalence
relation $\sim_{\cal F}$ is given in Figure~\ref{fig Cfk ef}.\\

\begin{figure}[H]
  \centerline{
	\begin{tikzpicture}[node distance=2.5cm,bend angle=30,transform shape,scale=1]
	  \node[accepting,state,rounded rectangle] (eps) {$\{C_{\varepsilon^1}(\overline{E})\}$};
	  \node[state, above of=eps,rounded rectangle,node distance=1.5cm] (h12) {$\{C_{f_1^1}(\overline{E}),C_{h_2^1}(\overline{E})\}$}; 
      \node[state, below of=eps,rounded rectangle,node distance=3.5cm] (g13) {$\{C_{g_3^1}(\overline{E})\}$};
      \node[state, right of=g13,rounded rectangle,node distance=3.5cm] (g23) {$\{C_{g_3^2}(\overline{E})\}$};
	  \node[state, left of=g13,rounded rectangle,node distance=3.5cm] (h15) {$\{C_{f_4^1(}\overline{E}),C_{h_5^1}(\overline{E})\}$};
	  \draw (eps) ++(-2cm,0cm) node {$b$}  edge[->] (eps);  
	  \draw (h12) ++(2cm,0cm) node {$b$}  edge[->] (h12); 
	  \draw (h15) ++(0cm,-1cm) node {$b$}  edge[->] (h15);  
	  \draw (g23) ++(2cm,0cm) node {$a$}  edge[->] (g23);  
	  \draw (g13) ++(0cm,-1cm) node {$b$}  edge[->] (g13);    
      \path[->]
		%
		(h12) edge[->,loop,above] node {$f,h$} ()
		(h12) edge[->,below right] node {$f,h$} (eps)
	    %
		(h15) edge[->, in=147,out=-147,loop,left] node {$f,h$} ()	
		(h15) edge[->,above left] node {$f,h$} (eps)		
		(h15) edge[->,above left] node {$f,h$} (g13)		
		%
	  ;
      \draw (eps) ++(1.75cm,-0.75cm)  edge[->,in=0,out=90] node[above right,pos=0] {$g$} (eps) edge[] (g13) edge[] (g23);  
      \draw (g13) ++(1.5cm,1cm)  edge[->,in=90,out=145] node[above right,pos=0] {$g$} (g13) edge[] (g13) edge[] (g23);
    \end{tikzpicture}
  }
  \caption{The Automaton ${\cal C_{\E}}\diagup_{\sim_{\cal F}}$.}
  \label{fig Cfk ef}
\end{figure}
		 
The quotient automaton of the automaton ${\cal C_{\E}}\diagup_{\sim_{\cal F}}$ by the equivalent relation $\sim_e$ is given in Figure~\ref{fig Cfk efe}.
The number of states is $|Q|=4$ and the number of transition rules is $|\Delta|=14$.

\begin{figure}[H]
  \centerline{
	\begin{tikzpicture}[node distance=2.5cm,bend angle=30,transform shape,scale=1]
	  \node[accepting,state,rounded rectangle] (eps) {$\{C_{\varepsilon^1}(\overline{E})\}$};
      \node[state, below of=eps,rounded rectangle,node distance=3.5cm] (g13) {$\{C_{g_3^1}(\overline{E})\}$};
      \node[state, right of=g13,rounded rectangle,node distance=3.5cm] (g23) {$\{C_{g_3^2}(\overline{E})\}$};
	  \node[state, below left of=eps,rounded rectangle] (h15) {$\{C_{f_1^1}(\overline{E}),C_{h_2^1}(\overline{E}),C_{f_4^1(}\overline{E}),C_{h_5^1}(\overline{E})\}$};
	  \draw (eps) ++(0cm,1cm) node {$b$}  edge[->] (eps);  
	  \draw (h15) ++(0cm,-1cm) node {$b$}  edge[->] (h15);  
	  \draw (g23) ++(2cm,0cm) node {$a$}  edge[->] (g23);  
	  \draw (g13) ++(0cm,-1cm) node {$b$}  edge[->] (g13);    
      \path[->]
		%
	    %
		(h15) edge[->,loop,above] node {$f,h$} ()	
		(h15) edge[->,below right] node {$f,h$} (eps)		
		(h15) edge[->,below left] node {$f,h$} (g13)		
		%
	  ;
      \draw (eps) ++(1.75cm,-0.75cm)  edge[->,in=0,out=90] node[above right,pos=0] {$g$} (eps) edge[] (g13) edge[] (g23);  
      \draw (g13) ++(1.5cm,1cm)  edge[->,in=90,out=145] node[above right,pos=0] {$g$} (g13) edge[] (g13) edge[] (g23);
    \end{tikzpicture}
  }
  \caption{The resulting automaton.}
  \label{fig Cfk efe}
\end{figure}

\end{example}

\section{Conclusion}
In this paper we define and recall different constructions of tree automata from a regular expression. 
The different automata and their relations (quotient, isomorphism) defined in this paper are represented in Figure~\ref{Rel Automat},
where our extension of Garc\'ia \emph{et al.} construction is denoted by $\equiv_V$-NFA.

\begin{figure}[H]
  \centerline{
	\begin{tikzpicture}[node distance=2cm,bend angle=30,transform shape,scale=1.25]
      \node[state,rounded rectangle,inner sep=0.25cm] (kCC) {\textbf{k-C-Continuation}};
      \node[state, right of=kCC,rounded rectangle,inner sep=0.25cm,node distance=3.5cm] (kP) {\textbf{k-Position}};
      \node[state, below of=kCC,rounded rectangle,inner sep=0.25cm] (Eq) {\textbf{Equation}};
      \node[state, below of=kP,rounded rectangle,inner sep=0.25cm] (Fo) {\textbf{Follow}};
      \node[state, below right of=Eq,rounded rectangle,inner sep=0.25cm,node distance=2.47cm] (Res) {$\mathbf{\equiv_V}$\textbf{-NFA}};
      \path
        (kCC) edge[<->] (kP)
        (kCC) edge[->,left] node {$\sim_e$} (Eq)
        (kCC) edge[->,above right] node {$\equiv$} (Fo)
        (kP) edge[->,right] node {$\sim_{\mathcal{F}}$} (Fo)
        (kCC) edge[->,above right,in=105,out=-45] node {$\equiv_V$} (Res)
        (Fo) edge[->,dashed] (Res)
        (Eq) edge[->,dashed] (Res)
      ;
    \end{tikzpicture}
  }
  \caption{Relation Between the Automata.}
\label{Rel Automat}
\end{figure}

We have shown that the $k$-position automaton and the $k$-c-continuations automaton are isomorphic, and that both the equation automaton and the follow autolaton are different quotients of the $k$-position automaton.

Looking for reductions 
of the set of states, we applied 
 the algorithm by Garc\'{\i}a ~\emph{et al.}~\cite{garcia}
 which allowed us to compute an automaton
 the size of which 
 is bounded above by the size of the 
 smaller of the follow and the equation automata.

\bibliography{bibliograph}

\begin{thebibliography}{10}

\bibitem{antimirov}
Antimirov, V.M.:
\newblock Partial derivatives of regular expressions and finite automaton
  constructions.
\newblock Theor. Comput. Sci. \textbf{155}(2) (1996)  291--319

\bibitem{brzozowski}
Brzozowski, J.A.:
\newblock Derivatives of regular expressions.
\newblock J. ACM \textbf{11}(4) (1964)  481--494

\bibitem{ZPC1}
Champarnaud, J.M., Ziadi, D.:
\newblock From c-continuations to new quadratic algorithms for automaton
  synthesis.
\newblock IJAC \textbf{11}(6) (2001)  707--736

\bibitem{ZPC2}
Champarnaud, J.M., Ziadi, D.:
\newblock Canonical derivatives, partial derivatives and finite automaton
  constructions.
\newblock Theor. Comput. Sci. \textbf{289}(1) (2002)  137--163

\bibitem{automate1}
Comon, H., Dauchet, M., Gilleron, R., Jacquemard, F., Lugiez, D., Loding, C.,
  Tison, S., Tommasi, M.:
\newblock Tree automata techniques and applications.
\newblock Available on: {\url{http://www.grappa.univ-lille3.fr/tata}} (October
  2007)

\bibitem{mohri1}
Cortes, C., Haffner, P., Mohri, M.:
\newblock Rational kernels: Theory and algorithms.
\newblock Journal of Machine Learning Research \textbf{5} (2004)  1035--1062

\bibitem{garcia}
Garc\'{\i}a, P., L{\'o}pez, D., Ruiz, J., Alvarez, G.I.:
\newblock From regular expressions to smaller nfas.
\newblock Theor. Comput. Sci. \textbf{412}(41) (2011)  5802--5807

\bibitem{glushkov}
Glushkov, V.M.:
\newblock The abstract theory of automata.
\newblock Russian Mathematical Surveys \textbf{16} (1961)  1--53

\bibitem{yu}
Ilie, L., Yu, S.:
\newblock Follow automata.
\newblock Inf. Comput. \textbf{186}(1) (2003)  140--162

\bibitem{khorsi}
Khorsi, A., Ouardi, F., Ziadi, D.:
\newblock Fast equation automaton computation.
\newblock J. Discrete Algorithms \textbf{6}(3) (2008)  433--448

\bibitem{automate2}
Kuske, D., Meinecke, I.:
\newblock Construction of tree automata from regular expressions.
\newblock RAIRO - Theor. Inf. and Applic. \textbf{45}(3) (2011)  347--370

\bibitem{Ouali}
Laugerotte, {\'E}., Sebti, N.O., Ziadi, D.:
\newblock From regular tree expression to position tree automaton.
\newblock In Dediu, A.H., Mart\'{\i}n-Vide, C., Truthe, B., eds.: LATA. Volume
  7810 of Lecture Notes in Computer Science, Springer (2013)  395--406

\bibitem{mcnaughton60}
McNaughton, R., Yamada, H.:
\newblock Regular expressions and state graphs for automata.
\newblock IEEE Trans. on Electronic Computers \textbf{9} (1960)  39--47

\bibitem{cie}
Mignot, L., Sebti, N.O., Ziadi, D.:
\newblock An efficient algorithm for the equation tree automaton via the
  k-c-continuations.
\newblock In Beckmann, A., Csuhaj-Varj{\'u}, E., Meer, K., eds.: CiE. Volume
  8493 of Lecture Notes in Computer Science, Springer (2014)  303--313

\end{thebibliography}
\end{document}